\documentclass[showpacs,amsmath,amssymb,twocolumn,prx,superscriptaddress,10pt,aps,nofootinbib]{revtex4-2}

\usepackage[english]{babel}
\usepackage[utf8]{inputenc}
\usepackage[T1]{fontenc}

\usepackage{algorithm}
\usepackage[noend]{algpseudocode}
\usepackage{float} 

\usepackage[]{graphicx}
\usepackage{siunitx}
\usepackage{amsmath,amssymb,amsthm,mathrsfs,amsfonts,dsfont,pifont}
\usepackage{subfigure, epsfig}
\usepackage{bm}
\usepackage{enumerate}
\usepackage[dvipsnames]{xcolor}
\usepackage{color}
\usepackage{fancybox, graphicx}
\usepackage[skins]{tcolorbox}
\usepackage{multirow}
\usepackage{mathtools}
\usepackage{thmtools}
\usepackage{thm-restate}
\usepackage{makecell}
\usepackage{mleftright}\mleftright 

\DeclarePairedDelimiter\bra{\langle}{|}
\DeclarePairedDelimiter\ket{|}{\rangle}
\DeclarePairedDelimiterX\braket[2]{\langle}{\rangle}{#1 \delimsize\vert #2}
\DeclarePairedDelimiterX\ketbra[2]{| }{|}{#1 \delimsize\rangle\!\delimsize\langle #2}
\newcommand{\vertiii}[1]{{\left\vert\kern-0.25ex\left\vert\kern-0.25ex\left\vert #1 
		\right\vert\kern-0.25ex\right\vert\kern-0.25ex\right\vert}}

\usepackage[colorlinks]{hyperref}
\hypersetup{
	draft=false,
	colorlinks=true,
	citecolor = {blue},
    linkcolor=blue,
    filecolor=magenta,      
    urlcolor=cyan,
}
\usepackage{cleveref}

\usepackage{tikz}
\usepackage{quantikz}

\usepackage{physics}   

\newcommand{\cL}{\mathcal{L}}

\newcommand{\bq}{\mathbf{q}}
\newcommand{\bp}{\mathbf{p}}

\newcommand{\bbI}{\mathbb{I}}
\newcommand{\bbR}{\mathbb{R}}

\newcommand{\vspan}{\operatorname{span}}

\newcommand{\rS}{\boldsymbol{\mathsf{S}}}
\newcommand{\rE}{\boldsymbol{\mathsf{E}}}
\newcommand{\rA}{\boldsymbol{\mathsf{A}}}
\newcommand{\rB}{\boldsymbol{\mathsf{B}}}
\newcommand{\rqq}{\boldsymbol{\mathsf{q}}}
\newcommand{\qft}{\mathrm{QFT}}
\newcommand{\qpe}{\mathrm{QPE}}
\newcommand{\prepf}{\text{Prep } f}
\newcommand{\ri}{\mathrm{i}}

\usepackage{ifdraft}
\ifdraft{
 \newcommand{\authnote}[3]{{\color{#3} {\bf  #1:} #2}}
 }{
 \newcommand{\authnote}[3]{}
 }

\newtcolorbox[auto counter]{tbox}[2][]{%
	enhanced, float=hbt, drop fuzzy shadow southeast,
	colback=white!5!white, colframe=white!30!black,
	width= .97\columnwidth,sharp corners,boxrule=0.8pt,
	title={#2}, #1
}

\newtheorem{theorem}{Theorem}
\newtheorem{lemma}{Lemma}

\newtheorem{remark}[theorem]{Remark}
\newtheorem{proposition}{Proposition}

\newtcolorbox{codebox}{enhanced, width=.95\columnwidth, halign = flush left, drop fuzzy shadow southeast, boxrule=0.4pt, sharp corners, colframe=black, colback=white}

\begin{document}

\title{Convergence monitoring of quantum Gibbs samplers}

\author{Nikolaos Louloudis}
\altaffiliation{These authors contributed equally.}
\affiliation{Institute for Quantum Information, RWTH Aachen University, Aachen, Germany}
\email{nikolaos.louloudis@rwth-aachen.de}

\author{Ruben Ibarrondo}
\altaffiliation{These authors contributed equally.}
\affiliation{Department of Physical Chemistry, University of the Basque Country UPV/EHU, Apartado 644, 48080 Bilbao, Spain}
\affiliation{EHU Quantum Center, University of the Basque Country UPV/EHU, Apartado 644, 48080 Bilbao, Spain}
\email{ruben.ibarrondo@ehu.eus}

\author{Mikel Sanz}
\affiliation{Department of Physical Chemistry, University of the Basque Country UPV/EHU, Apartado 644, 48080 Bilbao, Spain}
\affiliation{EHU Quantum Center, University of the Basque Country UPV/EHU, Apartado 644, 48080 Bilbao, Spain}
\affiliation{Basque Center for Applied Mathematics (BCAM), Alameda de Mazarredo 14, 48009 Bilbao, Spain}

\author{Mario Berta}
\affiliation{Institute for Quantum Information, RWTH Aachen University, Aachen, Germany}
\affiliation{Department of Computing, Imperial College London, London, United Kingdom}

\date{\today}


\begin{abstract}
Recent progress in fully quantum Markov chain Monte Carlo methods enables efficient Gibbs-state sampling on quantum computers [Chen {\it et al.}, Nature 646, 561 (2025)]. Although rigorous worst-case bounds on mixing times remain largely inaccessible for classically intractable systems, experience from classical Monte Carlo suggests that convergence of relevant observables may nevertheless be rapid. This raises the practical question of how to diagnose convergence efficiently, i.e., with at most polynomial overhead. We propose a low-cost criterion for convergence monitoring that exploits the weak measurements inherent in quantum Gibbs samplers and their qubit-efficient variants [Ding {\it et al.}, arXiv:2508.05703 (2025)]. Our approach is based on the observation that, at thermal equilibrium, the net energy flow between system and environment vanishes and energy-exchange statistics satisfy a balance condition. This condition appears in the distribution of (quasi-)frequencies extracted from the weak-measurement record and we use it to construct a Hamiltonian-agnostic stopping criterion based solely on data already generated by the sampler. We provide a statistical analysis, along with numerical and analytical studies to understand its performance, assumptions, and limitations.
\end{abstract}

\maketitle


\section{Introduction}

\paragraph{Gibbs states.} Preparing Gibbs states is a central task in quantum simulation. For a Hamiltonian $H$ and inverse temperature $\beta$, the Gibbs state
\begin{align}\label{eq:gibbs-state}
    \rho_\beta = \frac{e^{-\beta H}}{\operatorname{Tr}(e^{-\beta H})}
\end{align}
determines equilibrium properties such as the free energy, internal energy, entropy, and response coefficients, and it plays a central role in the study of finite-temperature many-body systems \cite{Binder1987,landau2021guide,takahashi1996thermo,maldacena2003eternal,zhu2020generation}. Beyond physics, Gibbs distributions also appear in areas such as Bayesian inference \cite{gelfand2000gibbs}, Boltzmann machines \cite{ackley1985learning, amin2018quantum}, and semidefinite optimization \cite{arora2012multiplicative, brandao2017quantum, van2017quantum}. Classically, finite-temperature properties are commonly estimated via Markov chain Monte Carlo methods \cite{metropolis1953equation, robert1999monte, levin2017markov}, notably quantum Monte Carlo \cite{foulkes2001quantum} when suitable path-integral or stochastic representations exist.
While extremely successful, these methods can become ineffective in regimes affected by the sign problem \cite{Henelius2000,pan2022sign}, which motivates the search for quantum algorithms that prepare thermal states directly (in the spirit of \cite{feynman1982simulating}).

\paragraph{Quantum Gibbs samplers.} The most direct quantum analogue of MCMC is the quantum Metropolis algorithm \cite{temme2011quantum}, but its reliance on highly accurate energy estimation makes it costly in many settings. Recent breakthroughs in quantum Gibbs sampling instead follow a different route: they construct a Lindbladian whose stationary state is the desired Gibbs state and simulate the corresponding open-system dynamics efficiently on a quantum computer \cite{chen2025efficient,Chen2023b,Ding2024,Ding2025}. These constructions are inspired by physical thermalization and the Davies generator \cite{davies1974markovian,davies1976markovian}, but replace exact resolution of Bohr frequencies by a finite-resolution weighted Fourier transform, yielding an algorithmically efficient approximation. In this sense, quantum Gibbs samplers (QGSs) can be viewed as fully quantum analogues of continuous-time MCMC driven by physically motivated dynamics.

\paragraph{Mixing times.} A central question for such methods is their mixing behavior. Rigorous bounds on mixing times are already difficult in classical MCMC and remain even more challenging in the quantum setting, especially for systems of genuine physical interest. Existing worst-case guarantees are valuable, but they may substantially overestimate the time required for relevant observables to equilibrate in practice \cite{Temme2013,Kastoryano2013,Kastoryano2016,Capel2021,bardet2023rapid,ding2024polynomial,rouze2024optimal,rouze2025efficient,kochanowski2025rapid,tong2025fast,vsmid2025polynomial,vsmid2025rapid,bergamaschi2025quantum,Becker2026a,Becker2026b}. Moreover, if one restricts attention to local updates, one may expect the familiar critical-slowing-down phenomena near phase transitions \cite{wolff1989critical,sokal1991beat}. At the same time, experience with classical Monte Carlo suggests that practical performance is often much better than worst-case theory alone would indicate and is therefore assessed using empirical convergence diagnostics \cite{Vivekananda2019,HandbookMCMC2011}.

\paragraph{Convergence in practice.} Interestingly, MCMC practitioners are not necessarily in contact with the mathematical physics community proving rigorous mixing time results. Instead, they lean on empirical diagnostics, such as monitoring trace plots, summary statistics, or effective sample sizes, and use these to decide when sampling has become reliable enough for the quantities of interest \cite{Vivekananda2019,HandbookMCMC2011}. Transplanting this workflow to QGSs is, however, nontrivial. Direct measurements disturb the state, so monitoring an observable at intermediate times generally requires restarting the preparation from scratch and repeating the evolution many times for many monitoring times (or energy coherent measurements wait for the duration of the autocorrelation time \cite{Jiang2026}). Although conceptually simple, this strategy can incur a substantial overhead in both runtime and sampling cost. The practical issue is therefore not only how fast a QGS mixes, but also how to diagnose this efficiently.

\paragraph{Convergence monitoring.} In this work, we propose a low-cost diagnostic criterion that directly exploits information already generated within the QGSs of \cite{chen2025efficient} and their qubit-efficient variants \cite{Ding2025}. Our starting point is that these samplers are implemented through weak system-bath interactions whose ancilla outcomes encode jump events and their associated frequencies, or quasi-frequencies in the finite-resolution setting. At thermal equilibrium, the system is in detailed balance with its environment: the net energy flow vanishes on average, and the statistics of energy-exchange events obey a corresponding balance condition. In the ideal Davies limit, this is expressed directly in terms of Bohr frequencies. For practical algorithmic samplers, this means that we can predict structural properties of the quasi-frequency distribution from the design choices of the QGS and, crucially, independently of the system Hamiltonian, without direct measurements of the system state.

\begin{figure}
    \centering
    \includegraphics[width=1\linewidth]{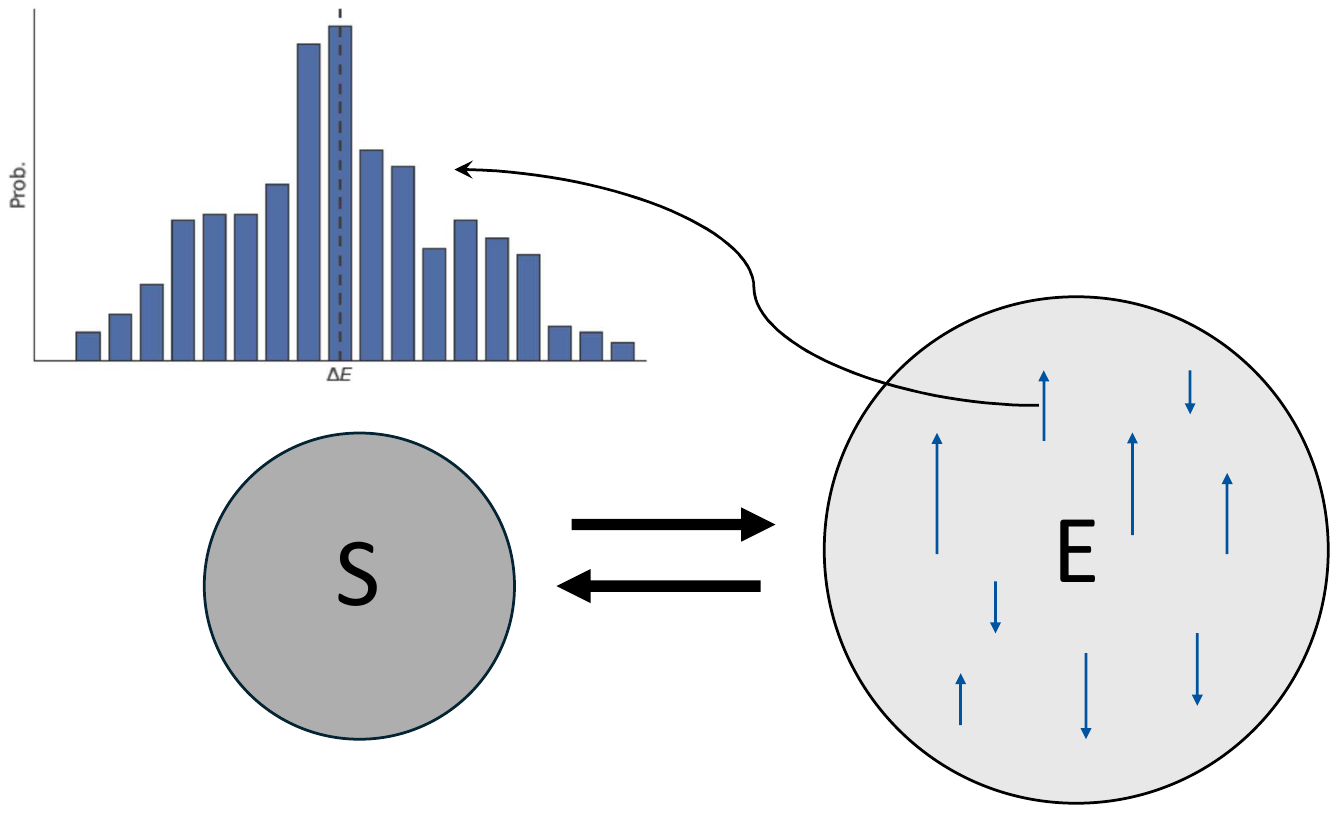}
    \caption{Davies limit intuition for the convergence monitoring scheme. Weak system--bath coupling induces transitions labeled by Bohr frequencies $\nu$. These transitions are inferred from the bath without directly measuring the system. At thermal equilibrium, detailed balance implies symmetry of the equilibrium transition weights, $M(\nu)=M(-\nu)$, so the transition profile is symmetric about $\nu=0$. Away from equilibrium an asymmetric profile signals a residual imbalance in the energy-exchange record. The algorithmic sampler of Sec. \ref{sec:qgs} replaces the Bohr frequencies by finite resolution quasi-frequencies. In Sec. \ref{sec:convergence} we derive the corresponding shifted symmetry of the quasi-frequencies which we use as the target of the stopping rule.}
    \label{fig:system_bath_frequency_measurments}
\end{figure}

\paragraph{Main result.} Based on this observation, we propose a simple stopping criterion for QGSs that relies solely on the weak-measurement record. The rule monitors the mean and third centered moment of the record and stops when a batch-means confidence ellipsoid is contained in a prescribed tolerance ellipsoid around their known equilibrium values. The criterion is low-cost, Hamiltonian-agnostic, and designed as a practical diagnostic rather than a worst-case certificate of global mixing. We analyze its statistical behavior, make explicit the assumptions under which it is informative, and provide some numerical simulations. In particular, Ddeterministic simulations show that the monitored moments co-relax with the energy and trajectory simulations show that the complete online rules detects the observed temperature-dependent slowing. On the analytical side, we discuss its physical meaning by relating the POVMs associated with (quasi-)frequencies to energy measurements. Namely, under a uniformity assumption on the jumps, the corresponding distributions are in bijective correspondence, implying that convergence of the observed frequency distribution entails convergence of the energy distribution. Our results indicate that convergence monitoring for quantum Gibbs samplers can be started from intrinsic measurement data, without requiring costly auxiliary measurements.

\paragraph{Related work.} Jiang \textit{et al.}~\cite{Jiang2026} and then Chen \textit{et al.}~\cite{Chen2026b} discuss interleaved detailed-balance measurements in the QGS framework of Ding \textit{et al.}~\cite{Ding2024}, with the goal of sampling thermal observables from a single Gibbs-sampling trajectory at a reduced cost governed by the autocorrelation time after an initial burn-in period. In Jiang \textit{et al.}, the explicit low-overhead construction is given for energy estimation and, more generally, for observables commuting with the Hamiltonian, while Chen \textit{et al.}~extend this framework to arbitrary non-commuting observables. As discussed in Jiang \textit{et al.}, an agnostic version can in practice also ``be used as an empirical method for certifying the convergence of QGSs.'' This approach incurs a small overhead and does not provide a direct statistical stopping analysis, whereas our results provide a dedicated convergence-monitoring criterion based solely on the weak-measurement record already produced by the sampler, at no additional cost. The single-trajectory framework is naturally compatible with both the KMS-detailed-balance samplers of Ding \textit{et al.}~\cite{Ding2024} and the algorithmic Lindbladian construction of Chen \textit{et al.}~\cite{chen2025efficient}. By contrast, for the qubit-efficient end-to-end scheme of Ding \textit{et al.}~\cite{Ding2025}, importing the current single-trajectory measurement machinery is less immediate. That algorithm derives its appeal from forward-only evolution with a single reusable ancilla qubit, whereas the explicit construction of Jiang \textit{et al.}~relies on Gaussian-filtered quantum phase estimation, so one would likely forfeit at least part of the no-QFT, no-extra-ancilla advantage. Our results, however, apply to the QGSs of Chen \textit{et al.}~\cite{chen2025efficient} and Ding \textit{et al.}~\cite{Ding2025}, but not to the KMS-sampler setting of Ding \textit{et al.}~\cite{Ding2024}. Another method to certify convergence for certain quantum Markov chains is by Stilck França~\cite{Stilck2018}, employing coupling-from-the-past techniques. Finally, Chen and Gily\'en~\cite{Chen2026} recently introduced efficient shadow tomography for QGSs, which allows them to sample multiple general observables from the Gibbs state with favorable scaling. While that work does not offer convergence-monitoring tools---as it assumes black-box thermal-state preparation---the crucial use of Hamiltonian-specific information seems promising for convergence diagnostics.

\paragraph{Structure.} The remainder of the paper is organized as follows. In Section \ref{sec:qgs}, we review the class of quantum Gibbs samplers considered here and the Lindbladian structure underlying their implementation. In Section \ref{sec:convergence}, we introduce the convergence criterion, its statistical analysis, and the numerical benchmarks. In Section \ref{sec:one-designs}, we present our analytical results on the relation between frequencies and energy distributions. We conclude with a brief outlook (Section \ref{sec:outlook}).


\section{Quantum Gibbs samplers}
\label{sec:qgs}

\paragraph{Standard framework.} We briefly explain the ingredients of the quantum Gibbs sampler needed for our convergence criterion. For more details see Appendix \ref{sec:long_overview_of_qgs}. For a Hamiltonian $H$ and inverse temperature $\beta$, the target is the Gibbs state in Eq.~\eqref{eq:gibbs-state}.
The algorithmic Gibbs samplers of \cite{chen2025efficient, Chen2023a, Chen2023b} construct an efficiently simulable Lindbladian whose fixed point is $\rho_\beta$. Starting from coupling operators $\{A_a\}_{a\in\mathcal{A}}$, with $\{A_a:a\in\mathcal{A}\}=\{A_a^\dagger:a\in\mathcal{A}\}$, one defines the filtered jump operators
\begin{equation}
\label{eq:Algorithmic_jumps}
\hat A_a(\omega)
=
\frac{1}{\sqrt{2\pi}}
\int_{\mathbb{R}} e^{itH}A_a e^{-itH} e^{-i\omega t} f(t)\,dt,
\end{equation}
where 
\begin{equation}\label{eq:Gaussian_window_time}
   f(t)=e^{-\sigma_E^2 t^2}\sqrt{\sigma_E\sqrt{2/\pi}},
\end{equation} 
is a Gaussian time window, in accordance with \cite[App. D]{Chen2023b}. The parameter $\sigma_E$ controls the frequency resolution and is chosen so that the jump operators in \eqref{eq:Algorithmic_jumps} are efficiently implementable via Hamiltonian simulation. We call the resulting $\omega\in\mathbb{R}$ quasi-frequencies. In the limit $\sigma_E\to 0$, the quasi-frequencies sharpen to the Bohr frequencies $\nu\in\mathcal{B}(H):=\{\epsilon_2-\epsilon_1:\epsilon_{1,2}\in\operatorname{spec}(H)\}$, and one recovers the Davies jump operators
\begin{equation}
\hat A_a(\nu)
=
\sum_{\epsilon_2-\epsilon_1=\nu}\Pi_{\epsilon_2}A_a\Pi_{\epsilon_1}.
\end{equation}
The corresponding algorithmic Lindbladian is
\begin{widetext}
\begin{equation}
\label{eq:Algorithmic_lindbladian}
\mathcal{L}_A(\rho)
=
-i[B,\rho]
+
\int d\omega \sum_{a\in\mathcal{A}} \gamma(\omega)
\left(
\hat A_a(\omega)\rho \hat A_a^\dagger(\omega)
-\frac{1}{2}\{\hat A_a^\dagger(\omega)\hat A_a(\omega),\rho\}
\right).
\end{equation}
\end{widetext}
The rates $\gamma(\omega)$ are chosen so that the dissipative part satisfies a KMS type detailed-balance condition with respect to $\rho_\beta$ \cite{Ding2024}
\begin{equation}
\label{eq:KMS_DB}
        \mathcal{L}_A^{\dagger}[\cdot]={\sqrt{\rho_\beta}}^{-1} \mathcal{L}_A\left[\sqrt{\rho_\beta} \cdot \sqrt{\rho_\beta}\right]{\sqrt{\rho_\beta}}^{-1}.
\end{equation}
Typical choices are
\begin{align}
\shortintertext{\textit{Metropolis}}
\gamma(\omega)
&=
\exp\left(-\beta \max\left(\omega+\frac{\beta\sigma_E^2}{2},0\right)\right)
\label{eq:Metropolis_transition_rate}
\\
\shortintertext{\textit{Gaussian}}
\gamma(\omega)
&=
\exp\left(-\frac{(\omega+\omega_\gamma)^2}{2\sigma_\gamma^2}\right),
\qquad
\beta=\frac{2\omega_\gamma}{\sigma_\gamma^2+\sigma_E^2}.
\label{eq:Gaussian_transition_rate}
\end{align}
More generally, one may also use the linear-combination-of-Gaussians family of \cite[Cor.~II.3]{Chen2023a}, which includes the shifted Metropolis and smooth Glauber-like filters of \cite[Prop.~II.4]{Chen2023a}; see App.~\ref{app:proof_quasi_frequency_symmetry}. With the corresponding Hermitian correction $B$, the full generator satisfies detailed balance and therefore leaves the Gibbs state invariant $\mathcal{L}_A(\rho_\beta)=0$. Under the usual assumption of ergodicity, the evolution converges to $\rho_\beta$.

\paragraph{Quantum simulation.} For the convergence criterion, the crucial point is that the simulation of $\mathcal{L}_A$ naturally produces weak-measurement data. Following \cite{Chen2023a}, we simulate the Hermitian part and the dissipative part $\mathcal{D}_A(\rho):=\mathcal{L}_A(\rho)+i[B,\rho]$, separately, incurring a small Trotter error. The former can be simulated using any known efficient Hamiltonian simulation algorithm. For the latter, we use the circuit of \autoref{fig:lindbladian_simulation}. The relevant block encoding satisfies
\begin{align}
\label{eq:block_encoding_jumps_algorithmic}
\bra{0}_{\rqq}\bra{0}_{\rB} \otimes \bbI_{\rA}\otimes\bbI_{\rS}\otimes \bbI_{\rE} U \ket{0}_{\rqq}\ket{0}_{\rB} \ket{0}_{\rA}\ket{\psi}_{\rS}\ket{0}_{\rE} = \nonumber \\ 
\sum_{a, \omega} \sqrt{\gamma(\omega)}  \ket{a}_{\rA}\ket{\omega}_{\rE} \hat{A}_a(\omega)\ket{\psi}_{\rS}
\end{align}
For time step $\delta$, the circuit implements the Euler step
\begin{equation}
\rho \mapsto e^{\delta\mathcal{D}_A}(\rho)+O(\delta^2)
=
\rho+\delta\mathcal{D}_A(\rho)+O(\delta^2).
\end{equation}
Standard implementations discard the ancilla outcomes after tracing them out. Here, we retain the quasi-frequency record and use its distribution as the central observable in Section \ref{sec:convergence}.

\paragraph{Qubit-efficient version.} A practical limitation of the above construction is the continuous parameterization. It requires a large time/frequency register and Fourier-transform machinery to resolve quasi-frequencies. Recent qubit-efficient variants \cite{Ding2025} alleviate these difficulties by replacing the continuous-frequency register with a single reusable ancilla qubit coupled weakly to the system. In each round, one samples a coupling operator and a bath frequency, prepares the ancilla in the corresponding thermal state, evolves under a weak system-bath interaction with Gaussian envelope, and then resets the ancilla. This same circuit also provides a natural monitored transition record: conditioned on an ancilla flip, one may retain the sampled frequency together with a sign determined by the flip direction, while rounds without a flip are discarded. For the Gaussian prior analyzed in Appendix~I of \cite{Ding2025}, the resulting signed-frequency histogram is described by a Gaussian KMS surrogate of the same transition-conditioned form studied in the next section. Thus, the distributional symmetry results and stopping-rule ideas developed there apply equally to this monitored qubit-efficient setting. We defer the derivation to Appendix~\ref{app:qubit_efficient_monitoring}.

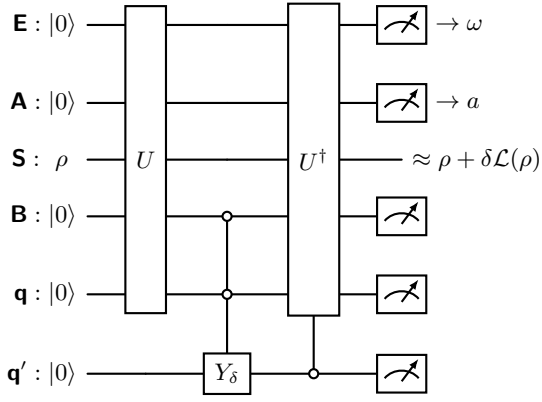
\begin{figure}
    \centering

    \begin{quantikz}
    	\lstick{$\rE: \ket{0}$} & \gate[5]{U} & & \gate[5]{U^{\dag}} & \meter{} \rstick{$\rightarrow \omega$} \\
        \lstick{$\rA: \ket{0}$} & & & & \meter{} \rstick{$\rightarrow a$}\\
        \lstick{$\rS:$ \makebox[0pt][l]{$\;\rho$}\phantom{$\ket{0}$}}  & & & & \rstick{$\approx \rho + \delta \cL(\rho)$}\\
        \lstick{$\rB: \ket{0}$} & & \octrl{1} & & \meter{}\\
        \lstick{$\rqq: \ket{0}$} & & \octrl{1} & & \meter{}\\
        \lstick{$\rqq' : \ket{0}$} & & \gate{Y_{\delta}} & \octrl{-1} & \meter{}
    \end{quantikz}
    
    \caption{Implementation of Lindbladian simulation via weak measurement. The unitary $U$ is the block encoding of the jumps in $\cL_\mathcal{A}$ as seen in Eq. \eqref{eq:block_encoding_jumps_algorithmic}. Its implementation is described in Fig.~\ref{fig:block_encoding_u}. There are two cases. If $\rqq'=0$, all ancillas are reset and the system remains in $\rho$. If $\rqq'=1$, the Lindbladian step is applied to $\rS$, and, conditioned on $\rB=\rqq=0$ the registers $\rA$ and $\rE$ return to $a$ and $\omega$, respectively; otherwise, the step is discarded.}

    \label{fig:lindbladian_simulation}
\end{figure}


\section{Convergence criterion}
\label{sec:convergence}

 In classical Gibbs samplers, convergence is often monitored through the trace plots of physically relevant observables, among which the energy plays a distinguished role: it is both informative and inexpensive to evaluate, since it is already built into the MCMC update rule. In the quantum setting, the analogous low-cost resource is the quasi-frequency record produced by the weak-measurement implementation of the sampler. In this section, we will show how this can be used to construct a practical stopping rule.

\paragraph{Distribution of quasi-frequencies.} Our stopping criterion is based on the quasi-frequencies observed in successful transition events of the weak-measurement circuit. Conditioned on such an event, the ancilla register \(\rE\) stores the quasi-frequency \(\omega\). The corresponding transition density is
\begin{equation}
\label{eq:mu_trans_main}
\mu_{\mathrm{trans}}(\omega,\rho)
=
\frac{
\sum_{a\in\mathcal{A}}
\gamma(\omega)\operatorname{Tr}\!\left[\hat A_a^\dagger(\omega)\hat A_a(\omega)\rho\right]
}{
\int d\omega'
\sum_{a\in\mathcal{A}}
\gamma(\omega')\operatorname{Tr}\!\left[\hat A_a^\dagger(\omega')\hat A_a(\omega')\rho\right]
}.
\end{equation}
Our central object of study is this probability density at equilibrium,
\begin{equation}
\label{eq:Pi_distribution_quasi_freq}
    \pi(\omega):=\mu_{\mathrm{trans}}(\omega,\rho_\beta).
\end{equation}
In the ideal Davies limit \(\sigma_E\to0\), the quasi-frequencies sharpen to exact Bohr frequencies and the equilibrium distribution is symmetric around \(0\), reflecting the vanishing net energy current between system and bath. For the Gaussian-filtered algorithmic Lindbladian, a closely related symmetry survives, but with a center shifted by the finite frequency resolution.

\begin{proposition}
\label{thm:quasi_frequency_symmetry}
Consider the algorithmic Lindbladian \eqref{eq:Algorithmic_lindbladian} with Gaussian window \eqref{eq:Gaussian_window_time}, and let \(\pi(\omega)=\mu_{\mathrm{trans}}(\omega,\rho_\beta)\) be the equilibrium quasi-frequency density. If the transition rate satisfies for $C:=-\frac{\beta\sigma_E^2}{2}$ that
\begin{equation}
\label{eq:rate_reflection_identity}
    \gamma(C-u)=e^{\beta u}\gamma(C+u)
    \qquad
    \forall u\in\mathbb R,
\end{equation}
then we have that
\begin{equation}
    \pi(C-u)=\pi(C+u)
    \qquad
    \forall u\in\mathbb R.
\end{equation}
In particular, we conclude that
\begin{equation}
    \langle \omega\rangle_\pi=C.
\end{equation}
\end{proposition}

The proof is given in detail  in Appendix~\ref{app:proof_quasi_frequency_symmetry}. The main idea is to rewrite the quasi-frequency density as Gaussian broadening of the Davies transition weights and then apply the usual Davies balance relation and the reflection identity of Eq. \eqref{eq:rate_reflection_identity} for the filter and rate.

Importantly, the reflection identity \eqref{eq:rate_reflection_identity} is satisfied by the Gaussian and Metropolis rates in \eqref{eq:Gaussian_transition_rate} and \eqref{eq:Metropolis_transition_rate}, as well as by the linear-combination-of-Gaussians family of \cite{Chen2023a}; see again Appendix~\ref{app:proof_quasi_frequency_symmetry}. Thus the equilibrium mean and symmetry of the sampled quasi-frequencies is known a priori from $\beta$ and $\sigma_E$, independently of the system Hamiltonian. The shift $C$ is a finite-resolution effect: In the Davies limit $\sigma_E\to0$, one recovers symmetry around $0$ and hence $\langle \nu\rangle_\pi=0$.


\paragraph{Algorithmic tool: stopping rules.}

Classical MCMC methodology includes both empirical convergence diagnostics and stopping rules \cite{Vivekananda2019}. Here we focus on stopping rules with quantitative error criteria, since these provide explicit tolerances for Monte Carlo error and are therefore more suitable for a rigorous convergence-monitoring scheme. Standard examples include fixed-width confidence-interval rules, relative fixed-width rules, multivariate fixed-volume rules, and effective-sample-size criteria \cite{jones2006fixed,flegal2008markov,flegal2010batch,vats2019multivariate}. In ordinary MCMC, such rules are formulated only in terms of precision, because the target expectation is unknown before the simulation. This changes in the present setting. \autoref{thm:quasi_frequency_symmetry} gives Hamiltonian-independent equilibrium constraints for the quasi-frequency record. We can therefore combine batch-means Monte Carlo uncertainty estimates with an additional target-aware accuracy requirement.

All sample indices below count accepted transition events. From the accepted
quasi-frequencies $\omega_i$, we monitor
\begin{equation}
\label{eq:main_monitored_observable}
    Y_i
    =
    \begin{pmatrix}
        \omega_i\\
        (\omega_i-C)^3
    \end{pmatrix},
\end{equation}
with known equilibrium target
\begin{equation}
    \theta^\star
    =
    \mathbb E_\pi[Y_i]
    =
    \begin{pmatrix}
        C\\
        0
    \end{pmatrix}.
\end{equation}
The first coordinate monitors the predicted equilibrium mean, while the second probes the first nontrivial odd centered moment beyond the mean. Restricting to these two coordinates keeps the diagnostic low-dimensional while testing both the location and the symmetry of the quasi-frequency record.

At each admissible checkpoint, let $\widehat\theta_M$ denote the mean over the selected accepted-event window, and let $\mathcal C_\alpha(M)$ be a confidence region constructed
using a non-overlapping batch-means estimate of the long-run covariance of the process
$(Y_i)$ \cite{vats2019multivariate}. Our stopping principle can be summarized as
\begin{equation}
\label{eq:main_unified_inclusion_principle}
\mathcal C_\alpha(M)\subseteq\mathcal T,
\end{equation}
where $\mathcal T$ is a prescribed tolerance region centered at the known target
$\theta^\star$. Thus, the rule stops when the confidence region is wholly contained inside the tolerance region. In one dimension, this reduces to
\begin{equation}
\label{eq:main_scalar_inclusion_prototype}
    |\widehat\theta_M-\theta^\star|+h_M\le P,
\end{equation}
which combines the observed displacement from the target with the Monte Carlo half-width.

For the main criterion, we use a covariance-aware ellipsoidal geometry. Let $\widehat{\Sigma}_{\mathrm{BM},M}$ be the non-overlapping batch-means estimate of the long-run covariance, and define
\begin{align}
    &\mathcal{C}_\alpha(M)\notag\\
    &=
    \left\{
        \theta\in\mathbb{R}^2:
        M
        \bigl(\widehat{\theta}_M-\theta\bigr)^\top
        \widehat{\Sigma}_{\mathrm{BM},M}^{-1}
        \bigl(\widehat{\theta}_M-\theta\bigr)
        \leq c_{\alpha,a}
    \right\},
    \label{eq:confidence_ellipsoid}
\end{align}
where $c_{\alpha,a}$ is the finite-batch critical value specified in Appendix \ref{app:stopping_rule_details}. For a positive-definite scale matrix $\Lambda$, define the target-centered tolerance ellipsoid
\begin{equation}
    \mathcal{T}_\Lambda(P_\Lambda)
    =
    \left\{
        \theta\in\mathbb{R}^2:
        (\theta-\theta^\star)^\top
        \Lambda^{-1}
        (\theta-\theta^\star)
        \leq P_\Lambda^2
    \right\}.
    \label{eq:tolerance_ellipsoid}
\end{equation}
The two covariance matrices play distinct roles. $\widehat{\Sigma}_{\mathrm{BM},M}$ determines the size and orientation of the Monte Carlo confidence ellipsoid, whereas $\Lambda$ specifies the marginal scale and geometry in which accuracy relative to
$\theta^\star$ is measured. We take $\Lambda$ either to be a fixed estimate of the stationary marginal covariance $\operatorname{Var}_{\pi}(Y_i)$ obtained from an independent calibration run, or an online estimate from a late accepted-event window. In the simulations below, we prefer the latter. Unlike the conventional fixed-volume criterion, this condition controls both the extent of the confidence ellipsoid and its displacement from the known equilibrium target.

For comparison, one may replace the ellipsoids by a Bonferroni confidence box and impose separate tolerances on the two monitored coordinates. Writing $\widehat{\theta}_M=(\widehat{m}_{1,M},\widehat{m}_{3,M})^\top$, with Bonferroni-adjusted batch-means half-widths $h_{1,M}$ and $h_{3,M}$, the coordinatewise inclusion rule stops when
\begin{align}
\label{eq:coordinatewise_inclusion}
    \left|\widehat{m}_{1,M}-C\right|+h_{1,M}\leq P_1, \\
    \left|\widehat{m}_{3,M}\right|+h_{3,M}\leq P_3.
\end{align}
This axis-aligned construction treats the two monitored moments separately, does not exploit their cross-covariance, and is not invariant under general linear recombinations of the moment coordinates. It can consequently be governed by the more weakly estimated third-moment coordinate. 

The complete online procedure is summarized in Algorithm \ref{alg:ellipsoidal_inclusion}. The Markov-chain central limit theorem, batch-means construction, finite-batch critical value, covariance windows, and asymptotic interpretation are collected in Appendix \ref{app:stopping_rule_details}.

\begin{algorithm}[H]
\caption{Ellipsoidal inclusion stopping from the accepted
quasi-frequency record}
\label{alg:ellipsoidal_inclusion}

\scriptsize
\begin{algorithmic}[1]

\Require tolerance $P_\Lambda>0$;
error level $\alpha$;
burn-in $N_0$;
check interval $m$;
target $\theta^\star=(C,0)^\top$;
source
$\mathsf{source}\in\{\mathrm{cal},\mathrm{online}\}$

\State $p\gets 2$, $N\gets 0$
\Comment{number of accepted transitions}

\Loop
    \State Perform one iteration of the Lindblad-evolution algorithm

    \If{no transition is accepted}
        \State \textbf{continue}
    \EndIf

    \State Let $\omega$ be the accepted quasi-frequency
    \State $N\gets N+1$
    \State
    $Y_N\gets
    \bigl(\omega,(\omega-C)^3\bigr)^\top$

    \If{$N\le N_0$ or $(N\bmod m)\neq 0$}
        \State \textbf{continue}
    \EndIf

    \State $N_{\mathrm{post}}\gets N-N_0$
    \State
    $a\gets\lfloor\sqrt{N_{\mathrm{post}}}\rfloor$

    \If{$a\le p$}
        \State \textbf{continue}
    \EndIf

    \State
    $b\gets\lfloor N_{\mathrm{post}}/a\rfloor$
    \State
    $M\gets ab$, \quad $s\gets N-M$

    \State
    $\displaystyle
    \widehat\theta_M
    \gets
    \frac{1}{M}\sum_{i=s+1}^{N}Y_i$

    \For{$k=0,\ldots,a-1$}
        \State
        $\displaystyle
        \overline Y_k
        \gets
        \frac{1}{b}
        \sum_{i=1}^{b}Y_{s+kb+i}$
    \EndFor

    \State
    $\displaystyle
    \widehat\Sigma_{\mathrm{BM},M}
    \gets
    \frac{b}{a-1}
    \sum_{k=0}^{a-1}
    (\overline Y_k-\widehat\theta_M)
    (\overline Y_k-\widehat\theta_M)^\top$

    \State
    $\displaystyle
    c_{\alpha,a}
    \gets
    \frac{p(a-1)}{a-p}
    F_{1-\alpha;\,p,\,a-p}$

    \State
    $\displaystyle
    A_M
    \gets
    \frac{c_{\alpha,a}}{M}
    \widehat\Sigma_{\mathrm{BM},M}$

    \If{$\mathsf{source}=\mathrm{cal}$}
        \State
        $\Lambda\gets\Lambda_{\mathrm{cal}}$
    \Else
        \State Estimate
        $\widehat\Lambda_{\mathrm{online}}$
        on the prescribed late accepted-sample window

        \State
        $\Lambda\gets
        \widehat\Lambda_{\mathrm{online}}$
    \EndIf

    \State
    $d_M\gets
    \widehat\theta_M-\theta^\star$

    \State
    $\displaystyle
    R_{\Lambda,M}^2
    \gets
    \sup_{\lVert u\rVert_2\le 1}
    \bigl(d_M+A_M^{1/2}u\bigr)^\top
    \Lambda^{-1}
    \bigl(d_M+A_M^{1/2}u\bigr)$

    \State
    $R_{\Lambda,M}
    \gets
    \sqrt{R_{\Lambda,M}^2}$

    \If{$R_{\Lambda,M}\le P_\Lambda$}
        \State
        \Return
        $(\widehat\theta_M,
          \widehat\Sigma_{\mathrm{BM},M},
          \Lambda,
          R_{\Lambda,M},
          M)$
    \EndIf
\EndLoop

\end{algorithmic}
\end{algorithm}

The rule certifies only the monitored equilibrium constraints at the chosen
tolerance. We treat it as a practical convergence monitor, not as a finite-time certificate of global mixing or of convergence of every observable.


\paragraph{Validation of the monitored moments}
\label{sec:monitored_moment_validation}
We first test whether the monitored quasi-frequency moments contain information about the equilibration of a physically relevant observable, namely the energy. To isolate the deterministic signal from finite-trajectory noise, we compare the exact evolution of the monitored moments with the deterministic relaxation of the energy. This comparison validates the signal used by the stopping rule and provides an offline physical calibration of the tolerance $P_\Lambda$. 

We consider the periodic nearest-neighbor and fully connected transverse-field Ising Hamiltonians
\begin{align}
    H_{\mathrm{I}}^{(n)}
    &=
    -\frac{1}{n}
    \left(
        J\sum_{i=1}^{n}Z_iZ_{i+1}
        +
        h\sum_{i=1}^{n}X_i
    \right),
    \label{eq:benchmark_ising}
    \\
    H_{\mathrm{FC}}^{(n)}
    &=
    -\frac{1}{n}
    \left(
        J\sum_{1\leq i<j\leq n}Z_iZ_j
        +
        h\sum_{i=1}^{n}X_i
    \right),
    \label{eq:benchmark_fcising}
\end{align}
where $Z_{n+1}=Z_1$. In both cases we set $J=1$, $h=0.5$, and use the $1/n$ normalization displayed above. The filtered jump operators are constructed from $\{X_i/\sqrt{2n},Z_i/\sqrt{2n}\}_{i=1}^n$, using the Metropolis rate in Eq. \eqref{eq:Metropolis_transition_rate} and set the frequency width to $\sigma_E=1$. Hence, the known equilibrium center of the monitored quasi-frequency distribution is $C=-\beta/2$. The benchmark grid is
\begin{equation}
    n\in\{3,5,7\},
    \qquad
    \beta\in\{0.5,1.0,1.5\}.
    \label{eq:benchmark_grid}
\end{equation}

For every case $c=(m,n,\beta)$, where $m \in \{ \text{I}, \text{FC}\}$ labels the model family, let $\rho_c(t)$ be the deterministic Lindblad evolution starting from the maximally mixed state, and let $\rho_{\beta,c}$ be the corresponding Gibbs state. Define
\begin{equation}
    \varepsilon_{E,c}(t)
    =
    \left|
        \operatorname{Tr}
        \left[
            H_c\bigl(\rho_c(t)-\rho_{\beta,c}\bigr)
        \right]
    \right|.
    \label{eq:energy_error}
\end{equation}
Let $\pi_{c,t}$ denote the accepted-transition quasi-frequency distribution induced by $\rho_c(t)$, and write
\begin{equation}
    \theta_c(t)=\mathbb{E}_{\pi_{c,t}}[Y],
    \qquad
    \theta_c^\star=
    \begin{pmatrix}
        C\\
        0
    \end{pmatrix}.
\end{equation}
Using the exact stationary marginal covariance $\Lambda_{\star,c}=\operatorname{Var}_{\pi_c}(Y)$ in equilibrium, where $\pi_c=\pi_{c,\infty}$, define the deterministic target distance
\begin{align}
    D_{\Lambda,c}(t)
    &=
    \left\|
        \theta_c(t)-\theta_c^\star
    \right\|_{\Lambda_{\star,c}^{-1}},
    \nonumber\\
    \text{equivalently,} \nonumber \\
    D_{\Lambda,c}(t)^2
    &=
    \bigl(\theta_c(t)-\theta_c^\star\bigr)^\top
    \Lambda_{\star,c}^{-1}
    \bigl(\theta_c(t)-\theta_c^\star\bigr).
    \label{eq:deterministic_target_distance}
\end{align}
This is the infinite-sample counterpart of the target displacement entering the online inclusion rule, expressed in the exact stationary marginal-covariance geometry.

Figure \ref{fig:monitored_moment_validation} shows a nearly monotone co-decay of the deterministic distance and the energy error for both model families, all three inverse temperatures, and all three system sizes. This provides empirical evidence that the monitored moments contain information about convergence of the energy observable. It does not yet test the complete stochastic stopping condition, which also contains the Monte Carlo uncertainty represented by the confidence ellipsoid. We incorporate that uncertainty in the trajectory benchmark below.

\begin{figure*}[t]
    \centering
    \includegraphics[width=\textwidth]{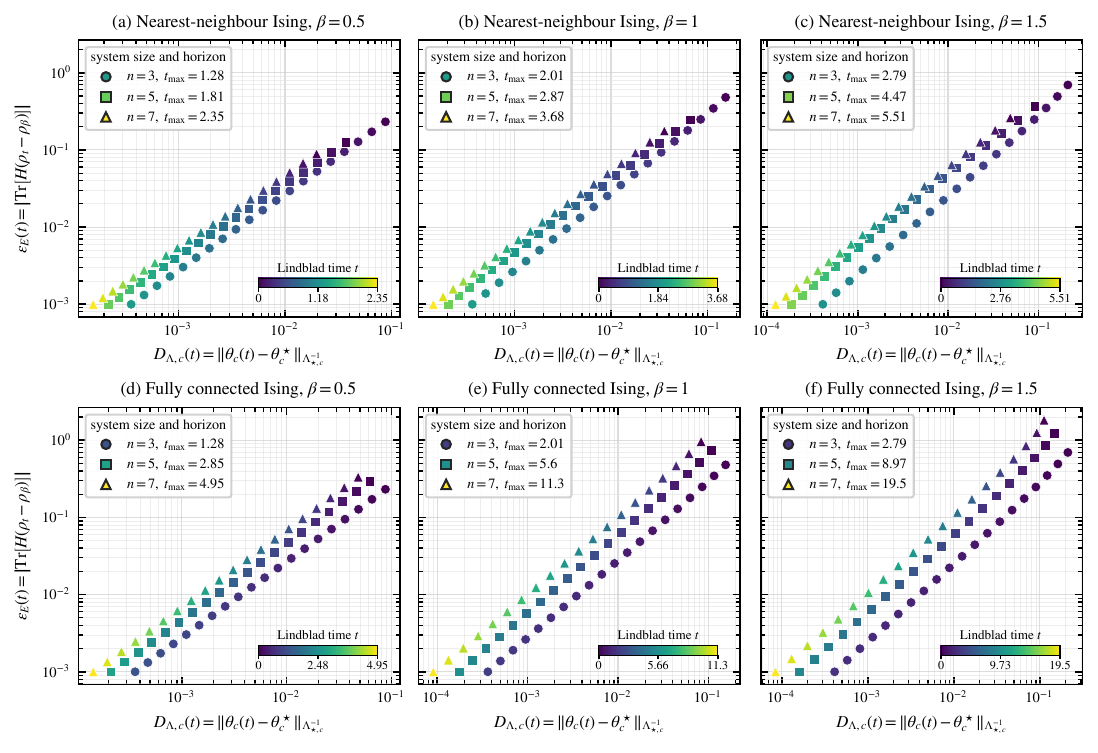}
    \caption{
        Relation between the deterministic target distance $D_{\Lambda,c}(t)$ and the energy error $\varepsilon_{E,c}(t)$ under deterministic Lindblad evolution. Rows correspond to the nearest-neighbour and fully connected Ising families, and columns to $\beta\in\{0.5,1.0,1.5\}$. Marker shape distinguishes $n\in\{3,5,7\}$, and color gives the physical Lindblad time. The legends report the actual terminal time $t_{\max}$ for each system size. The equilibration time increase with $n$ and $\beta$ and is also larger for the fully connected model. In all cases, we observe a co-decay which confirms that the quasi-frequency distribution can monitor the relaxation of the energy observable.
    }
    \label{fig:monitored_moment_validation}
\end{figure*}

\paragraph{Numerical benchmarking.}
\label{sec:numerical_benchmarking}
Having established at the deterministic level that the monitored moments respond to energy relaxation, we now test the complete online ellipsoidal inclusion rule of Algorithm \ref{alg:ellipsoidal_inclusion}. Specifically, we compare its inverse-temperature dependence with that of the deterministic energy-relaxation time over the benchmark grid in Eq.~\eqref{eq:benchmark_grid}. Unlike the comparison in \autoref{fig:monitored_moment_validation}, this test includes the Monte Carlo confidence region and the online estimation of the covariance geometry.

For each case $c$, we define the deterministic energy-relaxation reference time
\begin{equation}
    O_c
    =
    \inf\left\{
        t\geq0:
        \varepsilon_{E,c}(t)\leq10^{-3}
    \right\}.
    \label{eq:offline_reference_time}
\end{equation}
This can be used to calibrate the target-distance thresholds $P_\Lambda$. Table \ref{tab:stopping_tolerances} shows the chosen values. Within each model family, we hold $P_\Lambda$ fixed across $n$ and $\beta$, reflecting the practical setting in which its dependence on these parameters is not known in advance.

\begin{table}[t]
    \centering
    \caption{
        Target-distance tolerances used by the ellipsoidal inclusion rule. Each value of $P_\Lambda$ is fixed across $n$ and $\beta$ within the indicated model family.
    }
    \label{tab:stopping_tolerances}
    \begin{tabular}{lc}
        \hline\hline
        Model &$P_\Lambda$\\
        \hline
        I  &$2.0\times10^{-2}$\\
        FC &$1.3\times10^{-2}$\\
        \hline\hline
    \end{tabular}
\end{table}
For the stochastic simulation, we use the pure state trajectory method explained in Appendix \ref{app:trajectory_simulation}. Each trajectory starts from an independent Haar-random state, so the ensemble average is the maximally mixed state, and use a time step $\delta = 10^{-2}$ The point estimate $\widehat{\theta}_M$ and the long-run covariance estimate use the last $80\%$ of the accepted-transition record, whereas the online marginal covariance $\widehat{\Lambda}_M$ is estimated from its last $50\%$, with at least $200$ accepted transitions. The rule is evaluated after blocks of $10^3$ total sampler steps, once at least $50$ new accepted transitions have been collected. We use $95\%$ confidence regions throughout.

Let $\tau_c^{(j)}$ be the first-passage stopping time in trajectory $j$, measured in total sampler steps. For each case, we run $N_{\mathrm{runs}}=100$ independent trajectories and report
\begin{equation}
    S_c
    =
    \frac{1}{N_{\mathrm{runs}}}
    \sum_{j=1}^{N_{\mathrm{runs}}}
    \tau_c^{(j)}.
    \label{eq:mean_stopping_cost}
\end{equation}
Thus, $S_c$ estimates the expected computational cost of applying the stopping rule. Writing
\begin{equation}
    s_c^2
    =
    \frac{1}{N_{\mathrm{runs}}-1}
    \sum_{j=1}^{N_{\mathrm{runs}}}
    \left(\tau_c^{(j)}-S_c\right)^2,
\end{equation}
the error bars below are the ordinary $95\%$ Student-$t$ confidence intervals
\begin{equation}
    S_c
    \pm
    t_{0.975,N_{\mathrm{runs}}-1}
    \frac{s_c}{\sqrt{N_{\mathrm{runs}}}}.
    \label{eq:mean_stopping_cost_ci}
\end{equation}
These intervals quantify uncertainty in the estimated mean cost.

The deterministic quantity $O_c$ is expressed in Lindblad time, whereas $S_c$ counts circuit steps, so their raw magnitudes cannot be compared directly. Moreover, the stopping tolerance differs between the two model families. We therefore introduce one multiplicative scale for each model family $m$,
\begin{equation}
    a_m
    =
    \underset{c:\,m(c)=m}{\operatorname{median}}
    \frac{S_c}{O_c},
    \label{eq:model_specific_time_scale}
\end{equation}
and define
\begin{equation}
    \widetilde{S}_c
    =
    \frac{S_c}{a_{m(c)}}.
    \label{eq:model_calibrated_stopping_time}
\end{equation}
This calibration removes the conversion factor between sampler steps and Lindblad time, but it cannot change the ordering of cases or their relative dependence on $\beta$. The confidence limits for $\widetilde{S}_c$ are obtained by dividing both endpoints of Eq.~\eqref{eq:mean_stopping_cost_ci} by $a_{m(c)}$.

\begin{figure*}[t]
    \centering
    \includegraphics[width=\textwidth]{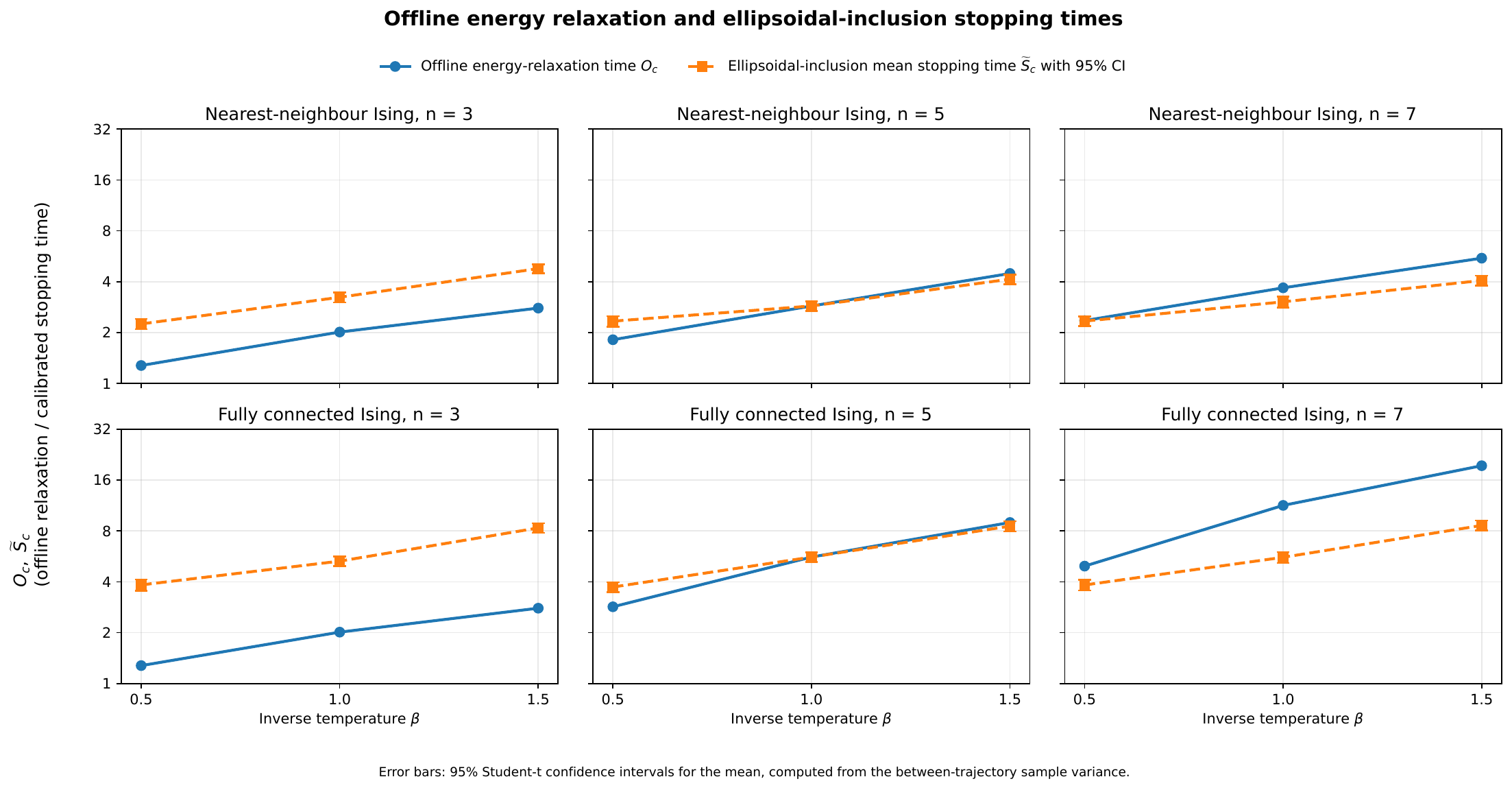}
    \caption{
        Inverse-temperature dependence of the deterministic energy-relaxation time $O_c$ and the model-calibrated mean stopping time $\widetilde{S}_c$ of the ellipsoidal inclusion rule. Rows correspond to the two model families and columns to $n\in\{3,5,7\}$. Error bars are $95\%$ Student-$t$ confidence intervals for the mean over $100$ independent trajectories. The deterministic reference carries no trajectory-sampling error bar. A base-two logarithmic vertical axis is used. We see that the scaling of $\widetilde{S}_c$ has the same qualitative behavior as the offline oracle $O_c$.
    }
    \label{fig:ellipsoidal_inclusion_beta_scaling}
\end{figure*}

\autoref{fig:ellipsoidal_inclusion_beta_scaling} shows that the ellipsoidal inclusion rule reproduces the temperature ordering of the offline energy-relaxation time in all six fixed-size series. In particular, both $O_c$ and $\widetilde{S}_c$ increase monotonically with $\beta$ for each model family and each $n$. Because
Eq.~\eqref{eq:model_specific_time_scale} supplies only one constant per model family, this ordering and the relative variation with $\beta$ are not imposed by the calibration. The agreement is qualitative rather than exact, and case-dependent differences remain.

We restrict this benchmark claim to the inverse-temperature dependence and do not infer system-size scaling from the present trajectory data. Across $n\in\{3,5,7\}$, the stochastic first-passage cost depends more weakly on $n$ than the deterministic reference time. This does not appear to result from an absence of size-dependent information in the exact monitored signal: in \autoref{fig:monitored_moment_validation}, the deterministic target distance co-decays with the energy error while the physical relaxation horizon increases with $n$. The most plausible explanation is that, at these small sizes and tolerances, the finite-sample precision requirement in $R_{\Lambda,M}$ dominates. This means that after the initial non-equilibrium relaxation, the trajectory must then continue until enough accepted transitions have accumulated to shrink the confidence ellipsoid. Admittedly, the present data do not establish this explanation conclusively.

\paragraph{Discussion.} The two numerical comparisons therefore support distinct conclusions. The deterministic comparison provides evidence that the monitored quasi-frequency moments contain information about convergence of the energy observable in the benchmark families. The trajectory comparison shows that this information remains useful in the complete online rule for detecting temperature-dependent slowing. Both of these results concern the energy observable. In the next section, we provide complementary analytical studies for the observed frequency-energy connection by identifying conditions under which quasi-frequency statistics retain information about the energy distribution.


\section{Relation between quasi-frequencies and energy distribution}
\label{sec:one-designs}

In Section \ref{sec:convergence} we introduced a practical stopping rule based on the quasi-frequency record and numerically assessed the relation between the monitored moments and energy relaxation. Here we give complementary analytical studies relating quasi-frequencies to the energy measurement itself. 

\paragraph{Special case: Davies limit.} The underlying intuition is clearest in the Davies limit: observing a Bohr frequency $\nu$ corresponds to an energy exchange of exactly that amount. This suggests that the statistics of observed frequencies should carry information about the energy distribution of the state. Our goal in this section is to make that statement precise for the algorithmic Lindbladian, where the observed quantities are quasi-frequencies $\omega$ rather than exact Bohr frequencies.

\paragraph{Recoverability between measurements.} Let $\{P_x\}_{x \in \mathsf{X}}$ and $\{Q_y\}_{y \in \mathsf{Y}}$ denote two POVMs, and assume they are related by a linear transformation described by a $|\mathsf{Y}| \times |\mathsf{X}|$ real matrix $G$ such that 
\begin{equation}
	Q_y = \sum_{x \in \mathsf{X}} G_{y x}\, P_x .
\end{equation}
For any state $\rho$, the probability distribution corresponding to the first POVM, $p_x = \Tr[P_x \rho]$, contains all the information about the second one, since $q_y = \Tr[Q_y \rho] = \sum_x G_{y x}\, p_x$, or, arranging the probabilities as column vectors, $\bq = G\, \bp$. Moreover, if $G$ admits a left pseudo-inverse $G^+$ satisfying $G^+ G = \bbI$, then we also have $\bp = G^+ \bq$, and in this sense the two POVMs are equivalent. Equivalently, this condition can be expressed as $\vspan\{P_x\}_{x \in \mathsf{X}} = \vspan\{Q_y\}_{y \in \mathsf{Y}}$. In this setting, $G$ represents a classical post-processing map between the POVMs, and the above condition ensures recoverability.

\paragraph{Our setting.} The reference measurement is the energy PVM $\{\Pi_{\epsilon}\}_{\epsilon\in\operatorname{spec}(H)}$. For the quasi-frequency record, define the operator density
\begin{equation}
    Q_{\omega} := \sum_{a\in\mathcal A} \gamma(\omega)\,\hat A_a(\omega)^{\dagger}\hat A_a(\omega),
\end{equation}
so that, for a small time step $\delta$, the quantity $\delta\,\Tr[Q_{\omega}\rho]\,d\omega$ is the probability of observing a transition with quasi-frequency in $[\omega,\omega+d\omega)$. Adjoining the decay event
\begin{equation}
    D := \bbI - \delta\int Q_{\omega}\,d\omega
\end{equation}
yields a bona fide POVM $\{D\}\cup\{\delta Q_{\omega}\}_{\omega\in\mathbb R}$ with output space $\{\Delta \} \cup \bbR$, where the symbol $\Delta$ denotes decay and the real values $\omega\in\bbR$ denote the observation of a transition with quasi-frequency $\omega$. The relevant information of this POVM is captured by the operator density $Q_{\omega}$, as can be seen from $\vspan \{D\} \cup \{\delta Q_{\omega}\}_{\omega\in\bbR} = \vspan \{\bbI\} \oplus \vspan \{Q_{\omega}\}_{\omega\in\bbR}$. Concentrating on $Q_{\omega}$ we can prove the desired connection between quasi-frequencies and energies. Note, the stopping rule of Sec. \ref{sec:convergence} uses the \emph{normalized} transition distribution $\delta \Tr[Q_{\omega} \rho] / (1 - \Tr[D \rho])$. The theorem below should therefore be read as analytical evidence for the physical content of the quasi-frequency record, rather than as a direct reconstruction formula for the practical estimator.

\paragraph{Jump operators.} To obtain a clean statement, we impose two simplifying assumptions. First, we take both the filter $f$ and the rate $\gamma$ to be Gaussian. Second, we assume that the couplings satisfy the \emph{uniform prior} condition
\begin{equation}
\label{eq:uniform_prior}
    \sum_{a\in\mathcal A} A_a^\dagger X A_a
    =
    \frac{\Tr(X)}{d}\,\bbI
\end{equation}
for every operator $X$, where $d=\dim\mathcal H$. This is the normalized unitary $1$-design condition, and it ensures that the jump family acts uniformly within each energy sector; an example being the family of Pauli strings.

\paragraph{General relation.} Extending the case of the Davies limit, the connection of the statistics of observed frequencies and the energy distribution of the state is as follows.

\begin{theorem}
\label{thm:energy_vs_frequencies}
Let $\{\Pi_{\epsilon}\}_{\epsilon\in\operatorname{spec}(H)}$ be the spectral projectors of the Hamiltonian $H$, and let $Q_{\omega}$ be the quasi-frequency operator density defined above. Assume that $f$ and $\gamma$ are Gaussian and that the coupling operators satisfy the uniform prior condition from Eq.~\eqref{eq:uniform_prior}. Then, we have that
\begin{equation}
    \vspan\{\Pi_{\epsilon}:\epsilon\in\operatorname{spec}(H)\}
    =
    \vspan\{Q_{\omega}:\omega\in\mathbb R\}.
\end{equation}
Equivalently, there exists a kernel $G$ from energies to quasi-frequencies with a left pseudo-inverse $G^+$ such that, for every state $\rho$,
\begin{align}
q(\omega,\rho)&=\sum_{\epsilon\in\operatorname{spec}(H)} G_{\omega\epsilon}\,p_{\epsilon}(\rho), \\
    p_{\epsilon}(\rho)&=\int G^+_{\epsilon\omega}\,q(\omega,\rho)\,d\omega,
\end{align}
where $p_{\epsilon}(\rho)=\Tr[\Pi_{\epsilon}\rho]$ and $q(\omega,\rho)=\Tr[Q_{\omega}\rho]$.
\end{theorem}

Theorem \ref{thm:energy_vs_frequencies} immediately extends to measurements whose outcome probabilities depend only on the energy distribution. More precisely a POVM,
\begin{equation}
\begin{aligned}
M_z
=
\sum_{\epsilon\in\operatorname{spec}(H)}
r(z\mid\epsilon)\,\Pi_\epsilon,
\\
r(z\mid\epsilon)\geq 0,
\qquad
\sum_{z\in\mathsf Z}r(z\mid\epsilon)=1.
\end{aligned}
\end{equation}
Such a measurement is a classical post-processing of the energy measurement, since
\begin{equation}
    \Pr(z\mid\rho)
    =
    \sum_{\epsilon\in\operatorname{spec}(H)}
    r(z\mid\epsilon)\,p_\epsilon(\rho).
\end{equation}
Consequently, because Theorem \ref{thm:energy_vs_frequencies} shows that the quasi-frequency statistics determine $p_\epsilon(\rho)$, they also determine the outcome statistics of every such measurement. For observables, the corresponding condition is $O=\sum_\epsilon o_\epsilon\Pi_\epsilon$. This is equivalent to $[O,H]=0$ when $H$ is nondegenerate.

\paragraph{Derivation.} We refer to Appendix~\ref{sec:proof_energy_vs_frequencies} for the full proof, but the main idea is to identify the map from energies $\epsilon$ to quasi-frequencies $\omega$ as the composition of two simpler transformations. The first maps energies to Bohr frequencies and is represented by a finite matrix whose full column rank follows from a triangular sub-matrix extracted from the ordered spectrum. The second maps Bohr frequencies to quasi-frequencies through a Gaussian kernel, and its full column rank follows from the linear independence of translated Gaussians. Together these two steps yield the desired left pseudo-inverse.

\paragraph{Conclusion.} This gives analytical support to the use of quasi-frequencies as convergence diagnostics, at least for energy incoherent observables. Under the assumptions of the theorem, the quasi-frequency record is not merely a heuristic proxy: it contains the full information of the energy distribution. In particular, if two states produce the same quasi-frequency statistics, then they must also produce the same energy distribution. Likewise, any finite-sample estimate of the quasi-frequency law induces a corresponding estimate of the energy law through \(G^+\). At the same time, the assumptions are deliberately convenient rather than maximally general, and the proof itself makes clear where they can be relaxed. This points naturally toward the more problem-tailored diagnostics discussed in the outlook.
\section{Outlook}
\label{sec:outlook}

We have outlined a stopping rule that requires no additional quantum operations or ancilla. It only records and classically post-process measurement outcomes already produced by the weak-measurement implementation of the sampler. In particular, the monitored version of the qubit-efficient sampler retains its single-reusable-ancilla architecture. This distinguishes our approach from the single-trajectory observable-estimation protocols of Jiang, Leng, and Lin and of Chen, Jiang, Li, and Ying \cite{Jiang2026,Chen2026b}, which insert additional Gibbs-preserving measurement channels and incur polylogarithmic overhead. This is particularly useful for qubit-efficient thermalization, where importing Gaussian-filtering weakens the no-QFT and low-ancilla advantages of the original construction. Nevertheless the different approaches are complementary. Whenever their implementations are compatible as in \cite{Chen2023b, Chen2023a}, the criterion introduced here can be used with loosened thresholds to determine the burn-in period, after which a single-trajectory protocol can be used to estimate observables with a resampling cost governed by an autocorrelation time rather than by the global mixing time. 

A broader research direction is to determine which techniques from classical MCMC can be ported to quantum Gibbs sampling while incurring at most polylogarithmic overhead. A recent example is quantum replica exchange \cite{chen2025quantum}. By coupling replicas at different temperatures, it gives a quantum analogue of parallel tempering and can overcome certain energy barriers. Other promising directions include quantum analogues of cluster or worm-type updates, flat-histogram and multicanonical methods, and adaptive selection of coupling operators and jump rates.

Finally, the criterion proposed here is deliberately Hamiltonian-agnostic. This makes it broadly applicable, but it also limits the information against which convergence can be tested. \autoref{thm:energy_vs_frequencies} suggests a systematic route toward Hamiltonian-informed diagnostics.
If spectral information is available, even approximately, the kernel mapping the energy distribution to the quasi-frequency distribution could be used to select features of the record that directly track energy observables, or to calibrate stopping tolerances in terms of a desired energy error. A more modest extension would use known locality, symmetries, degeneracies, or relevant observables to choose the monitored moments.


\paragraph*{Acknowledgments.}
We thank Stefan Wessel for providing an overview of classical Monte Carlo methods, Dootika Vats for helpful explanations of multivariate fixed-volume stopping rules, and Armando Bellante for discussions of the naive convergence-monitoring strategy. We also thank Fabian Hassler, Dominique Unruh, Julius A.~Zeiss, Tobias Rippchen, and Gereon Ko{\ss}mann for helpful discussions.

NL and MB acknowledge support from the German Federal Ministry of Research, Technology and Space (BMFTR) in the project QUantum Algorithms to SImulate MAny-body Physics (QuASi-MaP, Grant No. 13N17336), from the European Research Council Agreement No.~948139, from the Excellence Cluster – Matter and Light for Quantum Computing (ML4Q-2), and from the BMW Doctoral Program in Quantum Systems Integration through the ERS.
RI and MS thank support from the Basque Government BasQ initiative under the Q-STREAM project. They also acknowledge support from OpenSuperQ+100 (Grant No. 101113946) of the EU Flagship on Quantum Technologies, from Project Grant No. PID2024-156808NB-I00 and Spanish Ram\'on y Cajal Grant No. RYC-2020-030503-I funded by MI-CIU/AEI/10.13039/501100011033 and by “ERDF A way of making Europe” and “ERDF Invest in your Future”. RI also acknowledges the support of the Basque Government Ph.D. Grant No. PRE 2021-1-0102.

OpenAI’s ChatGPT, primarily using GPT-5.5 Pro, assisted with manuscript editing and the development and debugging of numerical and visualization code. All AI-assisted material was reviewed and verified by the authors, who take full responsibility for the results and conclusions.

\bibliography{lit}
\bibliographystyle{ultimate}

\newpage
\appendix

\tableofcontents

\section{Applications and broader context}

Preparing Gibbs states is a central task in quantum simulation. In physics, Gibbs states determine equilibrium quantities such as free energies, entropies, response coefficients, and finite-temperature phase diagrams, and they arise throughout condensed-matter theory, high-energy physics, and black-hole thermodynamics \cite{Binder1987,landau2021guide,takahashi1996thermo,maldacena2003eternal,zhu2020generation}. Beyond physics, Gibbs distributions also play important roles in Bayesian inference, Boltzmann machines, and optimization problems such as semi-definite programming \cite{gelfand2000gibbs,ackley1985learning,amin2018quantum,arora2012multiplicative,brandao2017quantum,van2017quantum}.

From the algorithmic perspective, the key challenge is not only preparing the Gibbs state but doing so efficiently. In both classical and quantum settings, the total cost is governed by mixing. Rigorous bounds on mixing times are already difficult in classical MCMC and remain even more challenging in the quantum setting \cite{Temme2013,Kastoryano2013,Kastoryano2016,Capel2021,bardet2023rapid,ding2024polynomial,rouze2024optimal,rouze2025efficient,kochanowski2025rapid,tong2025fast,vsmid2025polynomial,vsmid2025rapid,bergamaschi2025quantum}. Moreover, when the dynamics is based on local updates, one expects critical slowing down near phase transitions, just as in the classical case \cite{wolff1989critical,sokal1991beat}. This does not rule out more global update rules, but it does illustrate why practical convergence diagnostics remain valuable even when asymptotic guarantees are available.

In practice, convergence is rarely assessed solely through rigorous mixing-time bounds. Instead, practitioners typically monitor empirical diagnostics such as trace plots, summary statistics, Gelman--Rubin-type criteria, or effective sample sizes \cite{Vivekananda2019,HandbookMCMC2011,Gelman1992,geyer1992practical,flegal2008markov}. When convergence appears inadequate, they often deploy heuristic acceleration techniques such as tempering, parallel tempering, cluster moves, or adaptive proposals \cite{MarinariParisi1992,swendsen1986replica,HukushimaNemoto1996,SwendsenWang1987,Wolff1989,HaarioSaksmanTamminen2001,AndrieuThoms2008}. This gap between worst-case theory and practical diagnostics is part of the motivation for the convergence-monitoring viewpoint adopted in the main text.

The viewpoint of the main text is that weak-measurement records produced internally by quantum Gibbs samplers provide a natural source of such diagnostics. Rather than introducing additional destructive measurements of the system, one can exploit information already present in the ancilla outcomes of the Lindbladian simulation. The appendix material above explains the physical and algorithmic background behind this perspective, while Secs. \ref{sec:convergence} and \ref{sec:one-designs} develop its concrete consequences.


\section{Naive convergence monitoring}
\label{app:Naive Method}

A direct way to monitor the convergence of a quantum Gibbs sampler is to measure a suitable observable at a prescribed sequence of evolution times. Because such a measurement generally disturbs the system, the state must be prepared anew for every sample at every monitoring time. Thus, at a node $t_j$, one repeatedly prepares
\begin{equation}
    \rho(t_j) = e^{t_j\mathcal{L}}\!\left(\rho_0\right)
\end{equation}
and measures the chosen observable. Convergence may then be assessed, for example, by checking stability across several successive nodes or by fitting the late-time relaxation.

Let $N_{\mathrm{sh}}$ denote the number of independent preparations required at each node to obtain the desired statistical precision. The cumulative simulated evolution time of a grid $\{t_j\}_{j\in\mathcal{J}}$ is
\begin{equation}
    \mathcal{T}_{\mathrm{grid}}
    =
    N_{\mathrm{sh}}
    \sum_{j\in\mathcal{J}} t_j .
    \label{eq:naive_grid_cost}
\end{equation}
In the following, we isolate the dependence on the simulated evolution time.

Let $t_{\mathrm{mix}}(\varepsilon)$, or simply $t_{\mathrm{mix}}$ be the characteristic time at which a monitored observable converges to its equilibrium point up to tolerance $\varepsilon$. Here $t_{\mathrm{mix}}$ is used only for the cost comparison. The monitoring procedure does not know it in advance, and observing one quantity does not in general certify convergence of the full quantum state. After reaching the first grid point at or beyond $t_{\mathrm{mix}}$, suppose that the condition is required to remain satisfied for $l$ additional monitoring nodes for verification. We describe two methods to certify convergence.

\paragraph{Linearly spaced grid.}
Consider the grid
\begin{equation}
    t_j = j t_0,
    \qquad
    j=1,2,\ldots,
\end{equation}
with fixed spacing $t_0>0$, as shown in Fig.~\ref{fig:linear_spacing}. Define
\begin{equation}
    k
    =
    \left\lceil
        \frac{t_{\mathrm{mix}}}{t_0}
    \right\rceil ,
\end{equation}
so that
\begin{equation}
    (k-1)t_0
    <
    t_{\mathrm{mix}}
    \leq
    k t_0 .
\end{equation}
The first node guaranteed to lie at or beyond $t_{\mathrm{mix}}$ is therefore $k t_0$. Continuing through the $l$ additional verification nodes gives the final node $(k+l)t_0$. The total number of nodes is
\begin{equation}
    N_{\mathrm{lin}} = k+l
    =
    \Theta\!\left(
        \frac{t_{\mathrm{mix}}}{t_0}+l
    \right),
\end{equation}
and the cumulative evolution time is the arithmetic-series sum
\begin{align}
    \mathcal{T}_{\mathrm{lin}}
    &=
    t_0\sum_{j=1}^{k+l} j
    \nonumber\\
    &=
    \frac{t_0}{2}(k+l)(k+l+1)
    \nonumber\\
    &=
    \Theta\!\left(
        \frac{t_{\mathrm{mix}}^2}{t_0}
        +
        l\,t_{\mathrm{mix}}
        +
        l^2 t_0
    \right).
    \label{eq:linear_grid_cost}
\end{align}
Consequently, for fixed $t_0$ and fixed verification depth $l$, the linearly spaced strategy uses $\Theta(t_{\mathrm{mix}}/t_0)$ nodes and has cumulative evolution cost
\begin{equation}
    \mathcal{T}_{\mathrm{lin}}
    =
    \Theta\!\left(
        \frac{t_{\mathrm{mix}}^2}{t_0}
    \right).
\end{equation}

\paragraph{Geometrically spaced grid.} 
The quadratic dependence on $t_{\mathrm{mix}}$ can be avoided by using a geometrically spaced grid,
\begin{equation}
    t_j = 2^j t_0,
    \qquad
    j=0,1,\ldots,
\end{equation}
as illustrated in Figure \ref{fig:exponential_spacing}. Choose $k$ such that
\begin{equation}
    2^{k-1}t_0
    <
    t_{\mathrm{mix}}
    \leq
    2^k t_0 .
    \label{eq:exponential_grid_k}
\end{equation}
The first grid point at or beyond $t_{\mathrm{mix}}$ is then $2^k t_0$, and it satisfies
\begin{equation}
    t_{\mathrm{mix}}
    \leq
    2^k t_0
    <
    2t_{\mathrm{mix}} .
\end{equation}
After $l$ additional verification nodes, the final node is $2^{k+l}t_0$. The total number of nodes is
\begin{equation}
    N_{\mathrm{exp}}
    =
    k+l+1
    =
    \Theta\!\left(
        1+\log_2\!\frac{t_{\mathrm{mix}}}{t_0}+l
    \right),
\end{equation}
whereas the cumulative evolution time is
\begin{align}
    \mathcal{T}_{\mathrm{exp}}
    &=
    t_0\sum_{j=0}^{k+l}2^j
    \nonumber\\
    &=
    t_0\left(2^{k+l+1}-1\right)
    \nonumber\\
    &=
    \Theta\!\left(
        2^l t_{\mathrm{mix}}
    \right).
    \label{eq:exponential_grid_cost}
\end{align}
Thus, for a fixed verification depth $l$, exponential spacing reduces the node count to $\Theta(\log(t_{\mathrm{mix}}/t_0))$ and the cumulative evolution-time dependence to $\Theta(t_{\mathrm{mix}})$. The total cost is nevertheless not logarithmic in $t_{\mathrm{mix}}$: it is dominated by the longest fresh preparation. Moreover, each additional verification node doubles the next evolution time, producing the exponential dependence on $l$ in Eq.~\eqref{eq:exponential_grid_cost}. Both methods must still be multiplied by the number of samples $N_{\mathrm{sh}}$ required at every node and both require auxiliary measurements of freshly prepared states. By contrast, the convergence criterion introduced in the main text uses the weak-measurement record already generated during the Gibbs sampler evolution and does not require a separate family of direct measurements at predetermined times.

\begin{figure*}[t]
    \centering
    \subfigure[Linearly spaced grid]{%
        \includegraphics[width=1.00\columnwidth]{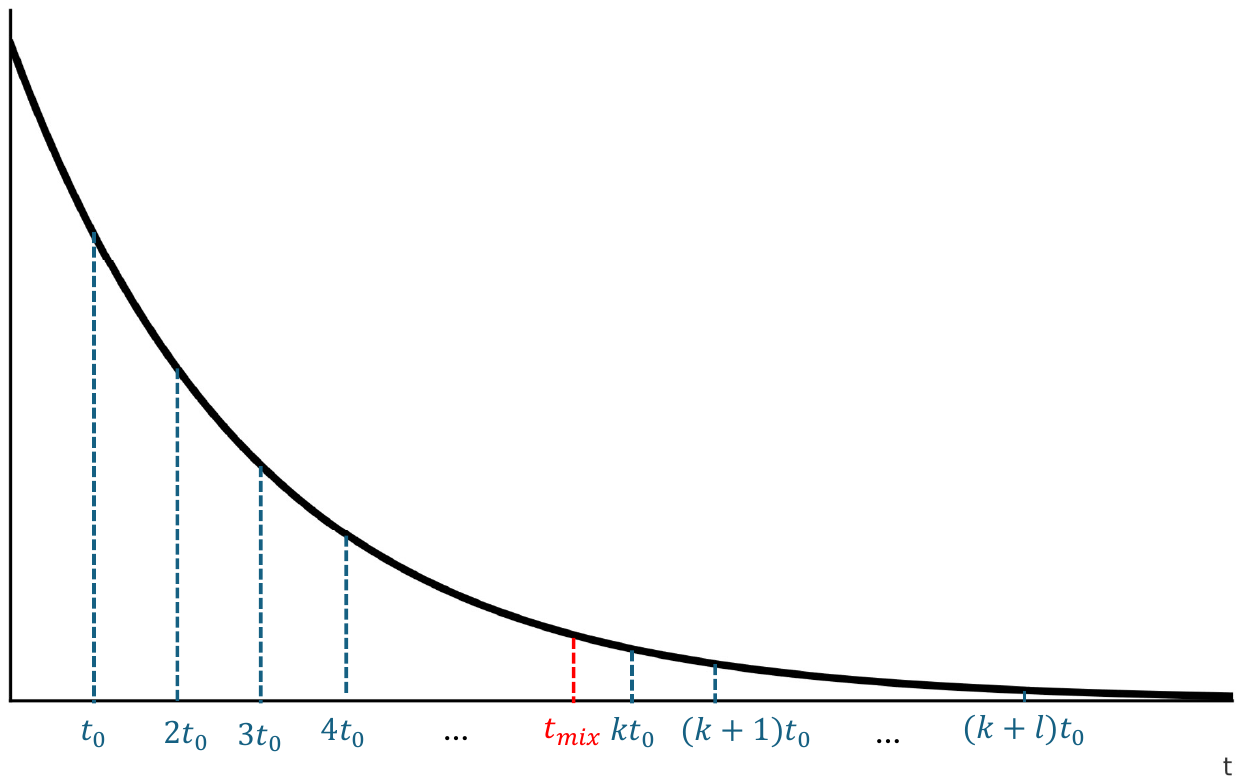}
        \label{fig:linear_spacing}}
    \hfill
    \subfigure[Geometrically spaced grid]{%
        \includegraphics[width=1.00\columnwidth]{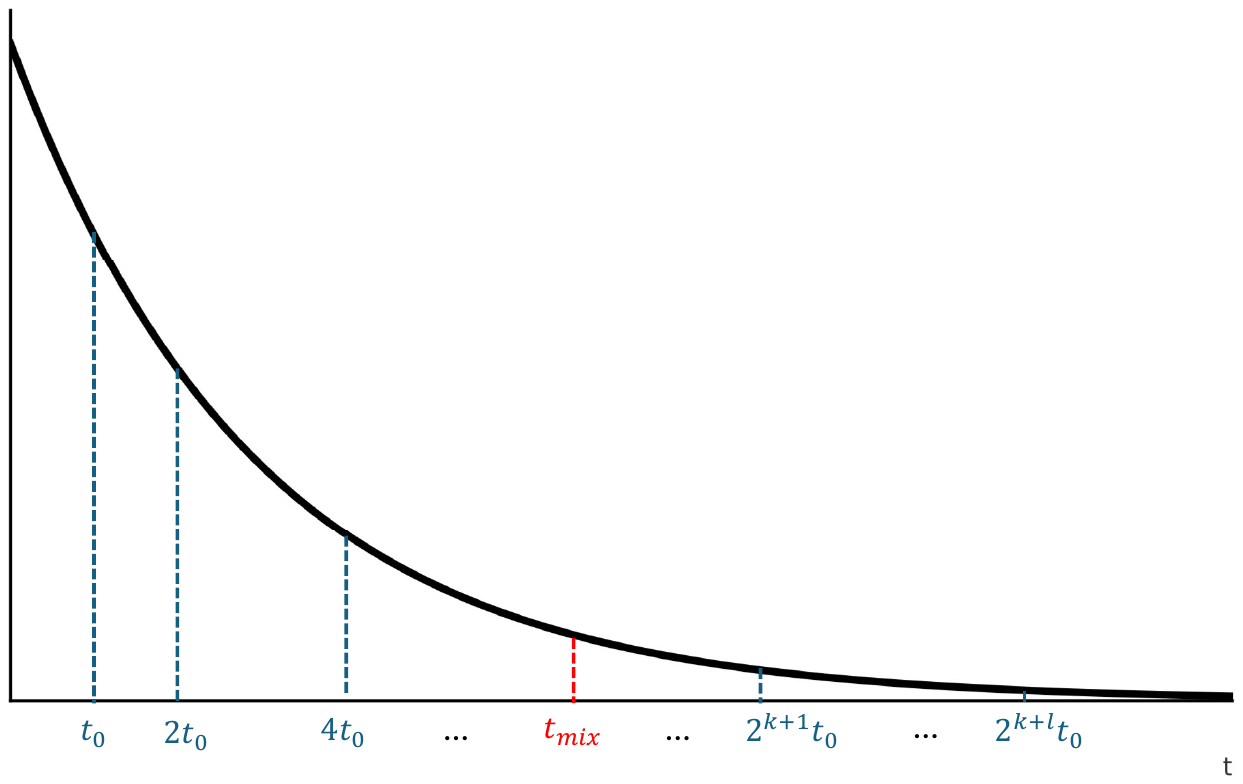}
        \label{fig:exponential_spacing}}
    \caption{Naive convergence monitoring by repeated state preparation and destructive measurement. The black curve represents the relaxation of an observable (for simplicity exponential), the vertical dashed lines denote monitoring times, and $t_{\mathrm{mix}}$ marks the unknown time at which the selected convergence tolerance is attained. For a fixed verification depth $l$, the linearly spaced grid uses $\Theta(t_{\mathrm{mix}}/t_0)$ nodes and cumulative evolution time $\Theta(t_{\mathrm{mix}}^2/t_0)$, whereas the geometrically spaced grid uses $\Theta(\log(t_{\mathrm{mix}}/t_0))$ nodes and cumulative evolution time $\Theta(t_{\mathrm{mix}})$. More generally, the verification overhead is $\Theta(2^l t_{\mathrm{mix}})$ for the geometrically spaced grid. All costs are multiplied by the number of independent samples required per node.}
    \label{fig:grid_spacing_for_measurement_times}
\end{figure*}


\section{Overview of Gibbs samplers}
\label{sec:long_overview_of_qgs}

In this appendix we expand on the construction of Sec. \ref{sec:qgs}. We first recall the classical detailed-balance picture behind Gibbs sampling and then review the Davies generator as a physical model of thermalization. We next explain why the exact jump operators of the Davies generator are not efficiently implementable and show how to construct algorithmic jump operators by Gaussian smearing of the Davies ones. Under the KMS detailed-balance condition, this leads to an algorithmically efficient Lindbladian whose evolution converges to the Gibbs state. Finally, we show how the associated weak-measurement evolution scheme produces the quasi-frequency distribution used in the main text.


\paragraph{Overview of classical MCMC.}
\label{subsec:classical_mcmc}

Classically, finite-temperature properties are often estimated by quantum Monte Carlo (QMC) methods \cite{foulkes2001quantum}, a family of Markov chain Monte Carlo (MCMC) schemes \cite{metropolis1953equation,levin2017markov,robert1999monte} adapted to path-integral and stochastic-series-expansion representations of the partition function. These methods are extremely successful, but in important regimes most notably frustrated or fermionic systems they can suffer from the sign problem \cite{Henelius2000,pan2022sign}. This is one of the main motivations for developing quantum algorithms that prepare Gibbs states directly, under the expectation that a quantum device should be naturally suited to simulating quantum many-body systems \cite{feynman1982simulating}.

Classical Markov chain Monte Carlo (MCMC) constructs a Markov chain whose unique stationary distribution is a desired target distribution. Let $\Omega$ be a finite state space and let $P(x\to y)$ denote the transition probabilities of a time-homogeneous Markov chain, satisfying $\sum_{y\in\Omega}P(x\to y)=1$ for all $x\in\Omega$. If $\mu_t$ denotes the distribution at time $t$, then
\begin{equation}
\label{eq:classical_markov_update}
    \mu_{t+1}(y)=\sum_{x\in\Omega}\mu_t(x)\,P(x\to y).
\end{equation}
A distribution $\pi$ is \emph{stationary} if it is a fixed point of the dynamics,
\begin{equation}
\label{eq:classical_stationary}
    \pi(y)=\sum_{x\in\Omega}\pi(x)\,P(x\to y)
    \qquad \text{for all } y\in\Omega.
\end{equation}
Under standard assumptions such as irreducibility and aperiodicity, the stationary distribution is unique and $\mu_t\to\pi$ as $t\to\infty$ for any initial distribution $\mu_0$.

A convenient sufficient condition ensuring that $\pi$ is stationary is \emph{detailed balance}:
\begin{equation}
\label{eq:classical_detailed_balance}
    \pi(x)\,P(x\to y)=\pi(y)\,P(y\to x),
    \qquad \text{for all } x,y\in\Omega.
\end{equation}
Summing \eqref{eq:classical_detailed_balance} over $x$ immediately gives \eqref{eq:classical_stationary}. This condition has a useful physical interpretation: at equilibrium there is no net probability current between any pair of states. To make contact with the energy-exchange statistics used in this paper, we introduce the one-way equilibrium flow
\begin{equation}
\label{eq:classical_oneway_flow}
    M(x,y):=\pi(x)\,P(x\to y),\quad x,y \in \Omega.
\end{equation}
Detailed balance is precisely the symmetry
\begin{equation}
\label{eq:classical_flow_symmetry}
    M(x,y)=M(y,x),
\end{equation}
i.e., the equilibrium flow is symmetric under time reversal.

When the goal is to sample from the Gibbs/Boltzmann distribution of the form
\begin{equation}
\label{eq:classical_gibbs_target}
    \pi_\beta(x)=\frac{e^{-\beta E(x)}}{Z(\beta)},
    \qquad
    Z(\beta)=\sum_{x\in\Omega} e^{-\beta E(x)},
\end{equation}
the detailed-balance condition \eqref{eq:classical_detailed_balance} becomes
\begin{equation}
\label{eq:Classical_KMS_condition}
    \frac{P(x\to y)}{P(y\to x)}=e^{-\beta(E(y)-E(x))}.
\end{equation}
A canonical construction is the Metropolis algorithm \cite{metropolis1953equation}. Given a symmetric proposal kernel $\Psi(x\to y)=\Psi(y\to x)$, at each step of the evolution, one proposes $y\sim\Psi(x\to\cdot)$ and accepts the move with probability
\begin{align}
\label{eq:metropolis_accept}
    \Gamma(x,y)
    &=
    \min\!\left\{
    1,\,
    \frac{\pi_\beta(y)}{\pi_\beta(x)}
    \right\} \nonumber \\
    &=
    \min\!\left\{1,e^{-\beta(E(y)-E(x))}\right\},
\end{align}
leading to a transition rule $P(x\to y)=\Gamma(x,y)\Psi(x\to y)$, for $x\neq y$, that satisfies \eqref{eq:classical_detailed_balance}.

This classical picture provides the right intuition for quantum Gibbs samplers. In the quantum setting, the Markov chain is replaced by a quantum Markov semigroup generated by a Lindbladian, and detailed balance is replaced by an appropriate KMS-type condition ensuring that the Gibbs state is stationary. The equilibrium-flow symmetry in \eqref{eq:classical_flow_symmetry} reappears as a symmetry of energy-exchange statistics associated with quantum jump processes.


\paragraph{Davies generator.}

Suppose that a system with Hamiltonian
\begin{equation}
    H=\sum_{\epsilon\in\operatorname{spec}(H)} \epsilon\,\Pi_\epsilon
\end{equation}
is weakly coupled to a thermal bath of inverse temperature $\beta$. In the weak-coupling, Born--Markov, and secular limit, the reduced dynamics are generated by a Lindbladian \cite{davies1974markovian,davies1976markovian} (for more modern derivations, see \cite[Sec.~XV]{lidar2019lecture} and \cite[Ch.~3.3]{breuer2002theory})
\begin{equation}
\label{eq:Lindblad_evolution}
    \dv{t}\rho=\mathcal{L}_D(\rho),
\end{equation}
which, in the form commonly used in the quantum algorithms literature \footnote{In fact, this follows from the original formulation of Davies by diagonalizing the Lindbladian and setting $\gamma_{\alpha \beta}(\nu)=\gamma (\nu)\delta_{\alpha \beta}$. Note, that when our goal is to construct a Gibbs Sampler, we have freedom on how to choose the transition rates, as long as they satisfy the KMS condition (or its analogue for the algorithmic Lindbladians of \cite{Chen2023b}), so this particular choice is not limiting for the current discussion.}, can be written as
\begin{widetext}
\begin{equation}
\label{eq:Davies_generator_conventional}
    \mathcal{L}_D(\rho)
    =
    -i[H,\rho]
    +
    \sum_{\nu\in\mathcal{B}(H)} \gamma(\nu)\sum_{a\in\mathcal{A}}
    \left(
    \hat A_a(\nu)\rho \hat A_a(\nu)^\dagger
    -\frac{1}{2}\{\hat A_a(\nu)^\dagger \hat A_a(\nu),\rho\}
    \right),
\end{equation}
\end{widetext}
where
\[
\mathcal{B}(H):=\{\epsilon_2-\epsilon_1:\epsilon_{1,2}\in\operatorname{spec}(H)\}
\]
is the set of Bohr frequencies and the rates $\gamma(\nu)$, arising from the bath fluctuations, satisfy the Kubo--Martin--Schwinger (KMS) condition
\begin{equation}
\label{eq:KMS_condition_first_davies}
    \gamma(\nu)=e^{-\beta\nu}\gamma(-\nu).
\end{equation}
The jump operators are
\begin{equation}
\label{eq:Jumps_of_Davies}
    \hat A_a(\nu)
    =
    \sum_{\epsilon_2-\epsilon_1=\nu}
    \Pi_{\epsilon_2}A_a\Pi_{\epsilon_1},
\end{equation}
where the $A_a$ are the coupling operators of the system to the bath.

From the definition, one immediately obtains
\begin{equation}
\label{eq:Dagger_relation_davies_jumps}
    \hat A_a(\nu)^\dagger=\hat A_a^\dagger(-\nu),
\end{equation}
and the decomposition $A_a=\sum_{\nu\in\mathcal{B}(H)}\hat A_a(\nu)$. Moreover, the jump operators are eigen-operators of the Heisenberg evolution generated by $H$:
\begin{align}
\label{eq:Commutation_relation_davies_jumps}
    e^{itH}\hat A_a(\nu)e^{-itH}&=e^{it\nu}\hat A_a(\nu), \\
    e^{\beta H}\hat A_a(\nu)e^{-\beta H}&=e^{\beta\nu}\hat A_a(\nu), \\
    [H,\hat A_a(\nu)]&=\nu \hat A_a(\nu).
\end{align}
Thus each $\hat A_a(\nu)$ mediates a transition between eigenspaces separated by energy $\nu$. Positive frequencies $\nu >0$ correspond to energy absorption, whereas negative ones $\nu <0$ correspond to emission. The KMS condition \eqref{eq:KMS_condition_first_davies} ensures that the ratio between them matches the Boltzmann factor. This implies that $\rho_\beta$ is stationary with respect to \eqref{eq:Lindblad_evolution}
\begin{equation}
\label{eq:Gibbs_is_steady_of_Davies}
    \mathcal{L}_D(\rho_\beta)=0.
\end{equation}
In the Davies setting, \eqref{eq:KMS_condition_first_davies} is equivalent to quantum detailed balance with respect to $\rho_\beta$. Under an ergodicity assumption, $\rho_\beta$ is then the unique steady state and the evolution converges to it.

To keep things aligned with the algorithmic Lindbladians used in Sec.~II, we will further assume that the family of coupling operators is closed under adjoint,
\begin{equation}
\label{eq:adjoint_assum}
    \{A_a:a\in\mathcal{A}\}=\{A_a^\dagger:a\in\mathcal{A}\}.
\end{equation}
This is automatically satisfied, for instance, when the $A_a$ are Hermitian.

The transition part of \eqref{eq:Davies_generator_conventional} already exhibits the symmetry that later motivates our stopping criterion. Define
\begin{equation}
    M(\nu):=\gamma(\nu)\sum_{a\in\mathcal{A}}
    \operatorname{Tr}\!\left[\hat A_a^\dagger(\nu)\hat A_a(\nu)\rho_\beta\right].
\end{equation}
This quantity measures the total equilibrium transition weight carried by frequency $\nu$.

\begin{lemma}
\label{lem:Davies_Mu}
Under the closure assumption \eqref{eq:adjoint_assum} and the KMS condition \eqref{eq:KMS_condition_first_davies}, we get
\begin{equation}
    M(\nu)=M(-\nu).
\end{equation}
\end{lemma}

\begin{proof}
Using cyclicity of the trace,
\begin{align}
    M(\nu)
    &=
    \gamma(\nu)\sum_{a\in\mathcal{A}}
    \operatorname{Tr}\!\left[\hat A_a(\nu)\rho_\beta \hat A_a^\dagger(\nu)\right].
\end{align}
From \eqref{eq:Commutation_relation_davies_jumps} we have
\[
\hat A_a(\nu)\rho_\beta=e^{\beta\nu}\rho_\beta \hat A_a(\nu),
\]
and therefore
\begin{align}
    M(\nu)
    &=
    \gamma(\nu)e^{\beta\nu}
    \sum_{a\in\mathcal{A}}
    \operatorname{Tr}\!\left[\rho_\beta \hat A_a(\nu)\hat A_a^\dagger(\nu)\right]
    \nonumber\\
    &=
    \gamma(-\nu)
    \sum_{a\in\mathcal{A}}
    \operatorname{Tr}\!\left[\rho_\beta \hat A_a(\nu)\hat A_a^\dagger(\nu)\right],
\end{align}
where we used \eqref{eq:KMS_condition_first_davies}. By \eqref{eq:Dagger_relation_davies_jumps},
\begin{equation}
\hat A_a(\nu)\hat A_a^\dagger(\nu)
=
\hat A_a^\dagger(-\nu)^\dagger\hat A_a^\dagger(-\nu),
\end{equation}
and since the set of couplings is closed under adjoint, the sum over $a$ may be relabeled. Hence
\begin{equation}
    M(\nu)
    =
    \gamma(-\nu)\sum_{a\in\mathcal{A}}
    \operatorname{Tr}\!\left[\hat A_a^\dagger(-\nu)\hat A_a(-\nu)\rho_\beta\right]
    =
    M(-\nu).
\end{equation}
\end{proof}

This lemma is the discrete Davies-limit statement of absorption-emission balance at thermal equilibrium. It is the starting intuition behind the quasi-frequency symmetry studied in Sec. \ref{sec:convergence}. Each jump operator $A_a(\nu)$ is associated with an energy change $\nu$ of the system and hence an opposite energy change with the environment. When the system is in the Gibbs state, we expect the frequencies to follow a symmetric distribution with respect to $0$, as the average energy exchange between the system and environment is $0$ at equilibrium, and additionally detailed-balance is equivalent to time-reversibility in the steady state. Furthermore, we can probe the properties of this distribution using information already manifest in the bath, without directly measuring the system state.


\paragraph{Algorithmic Lindbladian.}
\label{app:algorithmic_lindbladian}

The most direct quantum analogue of classical MCMC is the quantum Metropolis algorithm \cite{temme2011quantum}. More recent quantum Gibbs samplers instead follow a different strategy: they construct a Lindbladian whose stationary state is the Gibbs state and then simulate the corresponding open-system dynamics efficiently on a quantum computer \cite{chen2025efficient,Chen2023a,Chen2023b,Ding2024}. The algorithmic Lindbladians used in Sec.~II arise precisely in this way.

The starting point is a block encoding of a family of prior jump operators \(\{A_a\}_{a\in\mathcal A}\), chosen so that the resulting dynamics is ergodic and so that the set is closed under adjoint as in \eqref{eq:adjoint_assum}. Concretely, one assumes access to a unitary \(V\) such that
\begin{equation}
\label{eq:prior_jump_block_encoding}
    (\bra{0}_{\rB}\otimes \bbI_{\rA}\otimes \bbI_{\rS})\,
    V\,
    (\ket{0}_{\rB}\ket{0}_{\rA}\ket{\psi}_{\rS})
    =
    \sum_{a\in\mathcal A}\ket{a}_{\rA}A_a\ket{\psi}_{\rS}.
\end{equation}
The second primitive is controlled Hamiltonian evolution,
\begin{equation}
\label{eq:controlled_hamiltonian_evolution}
    \ket{t}_{\rE}\ket{\psi}_{\rS}
    \mapsto
    \ket{t}_{\rE}\,e^{-itH}\ket{\psi}_{\rS},
\end{equation}
whose cost depends on the chosen Hamiltonian-simulation method, sparsity and on the relevant evolution times.

To understand the role of these primitives, consider again the jump operators of \eqref{eq:Jumps_of_Davies}. They are obtained by resolving the coupling operator $A_a$ into components that connect eigenspaces of \(H\) separated by a Bohr frequency. From a quantum-algorithms perspective, the natural way to access these components is through idealized quantum phase estimation (QPE), which effectively performs an energy measurement of the system using controlled evolutions $e^{-itH}$. This viewpoint is illustrated in Fig.~\ref{fig:Anu_with_qpe}, which shows how a block encoding of the prior jumps can be converted into a block encoding of their operator Fourier transform, by sandwiching it with projections.

\begin{figure*} 
\centering 

\begin{quantikz} \lstick{$\rE: \ket{0}$} & \gate[2]{\qpe^{\dag}} && \gate[2]{\qpe} & \rstick{$\ket{\nu}$} \\ \lstick{$\rS:$ \makebox[0pt][l]{$\;\rho$}\phantom{$\ket{0}$}} & & \gate[3]{V} & & \rstick{$\hat{A}_a(\nu)\rho (\hat{A}_a(\nu))^{\dagger}$}\\ \lstick{$\rA: \ket{0}$} && && \rstick{$\ket{a}$} \\ \lstick{$\rB: \ket{0}$} & & & & \rstick{$\ket{0}$!} \end{quantikz} $$\equiv$$ \begin{quantikz} \lstick{$\rE: \ket{0}$} &\gate{\qft}& \ctrl{1} \wire[r][1]["t"{above, pos=0.05}]{q} &&\ctrl{1} \wire[r][1]["t"{above, pos=0.05}]{q}&\gate{\qft^{\dag}} & \rstick{$\ket{\nu}$} \\ \lstick{$\rS:$ \makebox[0pt][l]{$\;\rho$}\phantom{$\ket{0}$}} &&\gate{e^{-\ri t H}}& \gate[3]{V} &\gate{e^{\ri t H}}& & \rstick{$\hat{A}_a(\nu)\rho (\hat{A}_a(\nu))^{\dagger}$}\\ \lstick{$\rA: \ket{0}$} && & & & & \rstick{$\ket{a}$} \\ \lstick{$\rB: \ket{0}$} && & & & & \rstick{$\ket{0}$!} 
\end{quantikz}
\caption{Ideal Bohr-frequency resolution via phase estimation. Starting from a
block encoding \(V\) of the prior jumps, exact QPE resolves the coupling
operators into components \(\hat A_a(\nu)\) that connect energy sectors
separated by the Bohr frequency \(\nu\). The equivalent time-domain circuit
shows that \(\hat A_a(\nu)\) arises from the Fourier decomposition of
\(A_a(t)=e^{itH}A_ae^{-itH}\). This ideal picture motivates the filtered
construction used in the algorithmic sampler}
\label{fig:Anu_with_qpe} 
\end{figure*}
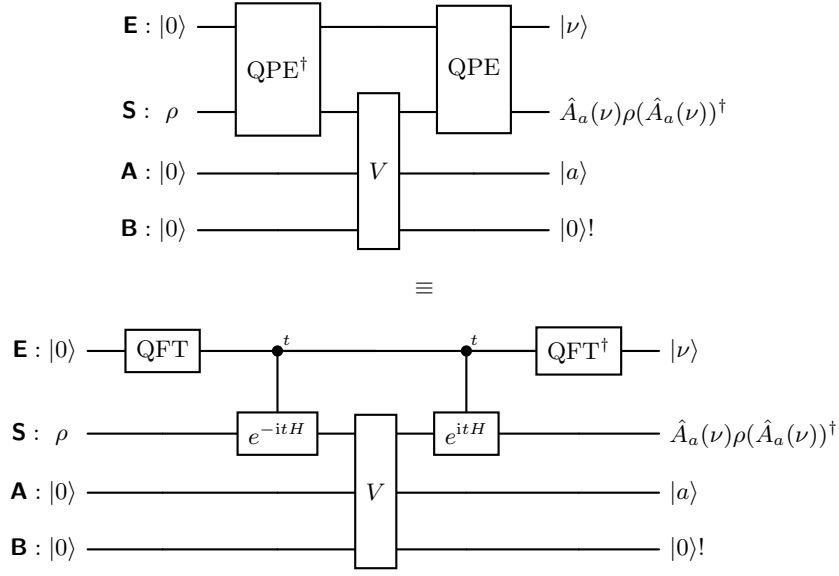

It is convenient to express this construction in the time domain, since this is the form directly used in the filtered implementation. Defining the Heisenberg-picture operator
\begin{equation}
A_a(t):=e^{itH}A_ae^{-itH},
\end{equation}
the spectral decomposition of $H$ gives
\begin{equation}
A_a(t)=\sum_{\nu\in\mathcal{B}(H)} e^{it\nu}\hat A_a(\nu),
\end{equation}
so that $\hat A_a(\nu)$ is precisely the coefficient of the Fourier mode $e^{it\nu}$. 
Thus the ideal frequency content of $A_a(t)$ is supported only on the discrete set of Bohr frequencies. Equivalently, one may encode this exact frequency decomposition by the operator-valued discrete measure
\begin{equation}
\mu_a(d\omega)
:=
\sum_{\nu\in\mathcal B(H)} \hat A_a(\nu)\,\delta_\nu(d\omega),
\end{equation}
or, in distributional notation,
\begin{equation}
\label{eq:distributional_fourier_Aa}
\mu_a(\omega)
=
\sum_{\nu\in\mathcal B(H)} \hat A_a(\nu)\,\delta(\omega-\nu),
\end{equation}
with
\begin{equation}
  \mathcal{F}\left[A_a\right](\omega)=\sqrt{2 \pi} \mu_a(\omega)=\sqrt{2 \pi} \sum_{\nu \in \mathcal{B}(H)} \hat{A}_a(\nu) \delta(\omega-\nu) .  
\end{equation}
This ideal picture also makes clear why the exact implementation is problematic. Perfectly resolving the delta peaks in \eqref{eq:distributional_fourier_Aa} requires arbitrarily fine frequency resolution, which in turn requires arbitrarily long evolution times and correspondingly arbitrarily large QPE registers. The algorithmic construction replaces this singular frequency measure by a smeared version. Let $g_{\sigma_E}$ be a frequency-resolution kernel, and define
\begin{equation}
\label{eq:smeared_jump_measure}
\hat A_a^{(\sigma_E)}(\omega)
:=
(\mu_a*g_{\sigma_E})(\omega)
=
\sum_{\nu\in\mathcal B(H)} \hat A_a(\nu)\,g_{\sigma_E}(\omega-\nu).
\end{equation}
In other words, each ideal delta peak at a Bohr frequency is replaced by a translated copy of the profile $g_{\sigma_E}$. When the filter is smooth and effectively band-limited in the frequency domain, the time-domain kernel has rapidly decaying tails, so the integral defining the filtered jump can be truncated with controlled error \cite{Ding2024}. A particularly convenient choice is Gaussian (see \cite[App.D]{Chen2023b} for an argument on why it is a unique choice in the case of a real filter),
\begin{equation}
g_{\sigma_E}(\omega)=\hat f_{\sigma_E}(\omega),
\end{equation}
where $f_{\sigma_E}$ is the time-domain Gaussian window and $\hat f_{\sigma _E}$ is the corresponding frequency-resolution kernel. Substituting into \eqref{eq:smeared_jump_measure} gives
\begin{align}
\hat A_a^{(\sigma_E)}(\omega)
&=
\sum_{\nu\in\mathcal B(H)}
\hat A_a(\nu)
\frac{1}{\sqrt{2\pi}}
\int
f_{\sigma_E}(t)e^{-i(\omega-\nu)t}\,dt \nonumber \\
&=
\frac{1}{\sqrt{2\pi}}
\int
f_{\sigma_E}(t)e^{-i\omega t}
\left(\sum_{\nu\in\mathcal B(H)} e^{it\nu}\hat A_a(\nu)\right)\,dt \nonumber \\
&=
\frac{1}{\sqrt{2\pi}}
\int
e^{itH}A_ae^{-itH}\,e^{-i\omega t}\,f_{\sigma_E}(t)\,dt.
\end{align}
Thus the algorithmic jump operators are precisely the Fourier transform of the Heisenberg-evolved operator after multiplication by the time window $f_{\sigma_E}$. This is the form directly implemented by the filtered-jump circuit: one prepares a superposition over times weighted by $f_{\sigma_E}(t)$, applies the forwrad and backward controlled evolutions around the prior jump block encoding $V$, and Fourier transforms the time register to obtain the broadened frequency label $\omega$.

For the Gaussian window
\begin{equation}
f_{\sigma _E}(t)=e^{-\sigma_E^2 t^2}\sqrt{\sigma_E\sqrt{2/\pi}},
\end{equation}
its Fourier transform is
\begin{equation}
\hat f_{\sigma_E}(\omega)
=
\frac{1}{\sqrt{\sigma_E\sqrt{2\pi}}}
\exp\!\left(-\frac{\omega^2}{4\sigma_E^2}\right).
\end{equation}
The width of  $\hat f _{\sigma_E}$ determines the achievable frequency resolution, while the spread of $f_{\sigma_E}$ determines the range of evolution times that must be implemented. Sharper frequency resolution therefore requires longer observation times. In practice, both the time register and the quasi-frequency register are discretized to finite precision and restricted to finite intervals; these additional errors can be controlled together with the truncation error of the filter.

The resulting generator takes the form of \eqref{eq:Algorithmic_lindbladian}. Enforcing the KMS detailed-balance \eqref{eq:KMS_DB} on the transition part gives the typical choices of transition rates. The Metropolis and Gaussian families in \eqref{eq:Metropolis_transition_rate} and \eqref{eq:Gaussian_transition_rate}. Enforcing the KMS-type detailed-balance on the rest of the Lindbladian, gives the corresponding Hermitian correction
\begin{equation}
    B=\frac{i}{2}\sum_{\nu\in\mathcal{B}(H)}
    \tanh\!\left(\frac{\beta\nu}{4}\right)R(\nu),
\end{equation}
where
\begin{equation}
    R=
    \int d\omega\, \gamma(\omega)\sum_{a\in\mathcal{A}}
    \hat A_a^\dagger(\omega)\hat A_a(\omega),
\end{equation}
with
\begin{equation}
    R(\nu):=\sum_{\epsilon_2-\epsilon_1=\nu}\Pi_{\epsilon_2}R\Pi_{\epsilon_1},
\end{equation}
the full generator satisfies detailed balance and leaves the Gibbs state invariant.

The Davies generator is recovered from this structure in an integrated sense. More precisely, for a family of windows such that $|\hat f_{\sigma_E}|^2$ forms an approximate identity as $\sigma_E\to0$, the diagonal contributions converge to the Davies rates at each Bohr frequency, while the off-diagonal contributions vanish after integration over $\omega$. This is the rigorous sense in which the algorithmic dissipator converges to the Davies dissipator as the frequency resolution is sharpened.

The above construction also admits a direct circuit implementation. Using the prior-jump block encoding $V$, the controlled Hamiltonian evolutions, a preparation of the time superposition weighted by $f$, and a Fourier transform on the time register, one obtains a unitary $U$ satisfying
\begin{align}
\label{eq:appendix_filtered_jump_block_encoding}
    &(\bra{0}_{\rqq}\bra{0}_{\rB}\otimes \bbI_{\rA}\otimes \bbI_{\rS}\otimes \bbI_{\rE})
    \,U\,
    (\ket{0}_{\rqq}\ket{0}_{\rB}\ket{0}_{\rA}\ket{\psi}_{\rS}\ket{0}_{\rE})
    \nonumber\\
    &\hspace{6em}=
    \sum_{a,\omega}
    \sqrt{\gamma(\omega)}\,
    \ket{a}_{\rA}\ket{\omega}_{\rE}\,
    \hat A_a(\omega)\ket{\psi}_{\rS}.
\end{align}
The construction is shown in Fig.~\ref{fig:block_encoding_u}. The register $\rE$ is first prepared in a superposition of times with amplitudes determined by the filter $f$; controlled forward and backward evolutions by $H$ implement the Heisenberg evolution around the prior-jump block encoding $V$; a Fourier transform maps the time register to the quasi-frequency register; and a controlled rotation weights the amplitudes by $\sqrt{\gamma(\omega)}$, with success flagged in the auxiliary qubit $\rqq$.

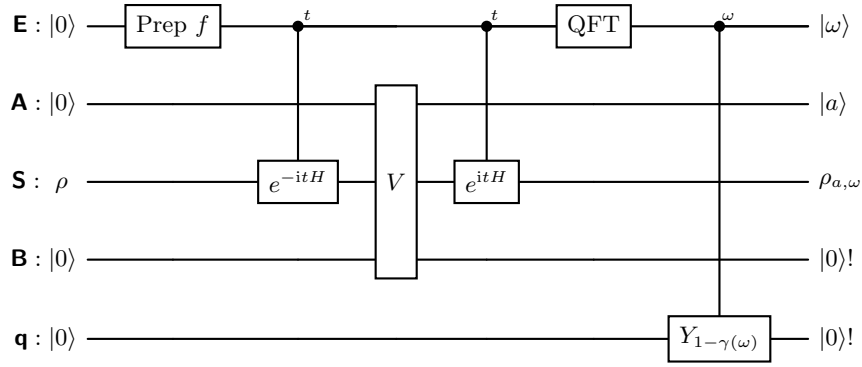
\begin{figure*} 
\centering 

\begin{quantikz} \lstick{$\rE: \ket{0}$} &\gate{\prepf}& \ctrl{2} \wire[r][1]["t"{above, pos=0.05}]{q} &&\ctrl{2} \wire[r][1]["t"{above, pos=0.05}]{q}&\gate{\qft} & \ctrl{4} \wire[r][1]["\omega"{above, pos=0.05}]{q} & \rstick{$\ket{\omega}$} \\ \lstick{$\rA: \ket{0}$} && & \gate[3]{V} & & & & \rstick{$\ket{a}$} \\ \lstick{$\rS:$ \makebox[0pt][l]{$\;\rho$}\phantom{$\ket{0}$}} &&\gate{e^{-\ri t H}}& &\gate{e^{\ri t H}}& & & \rstick{$\rho_{a,\omega}$}\\ \lstick{$\rB: \ket{0}$} && & & & & & \rstick{$\ket{0}$!}\\ \lstick{$\rqq: \ket{0}$} && & & & & \gate{Y_{1-\gamma(\omega)}} & \rstick{$\ket{0}$!} 
\end{quantikz} 

\caption{Construction of the block-encoding unitary $U$ used for the filtered jump operators in \eqref{eq:Algorithmic_jumps}. Here $(\bra{0}_{\rB}\otimes \bbI_{\rA}\otimes \bbI_{\rS})\,V\,(\ket{0}_{\rB}\ket{0}_{\rA}\ket{\psi}_{\rS}) = \sum_{a\in\mathcal A}\ket{a}_{\rA}A_a\ket{\psi}_{\rS}$ is the block encoding of the prior jumps. The register $\rE$ is prepared in the time superposition weighted by the filter $f$, the system undergoes controlled forward and backward evolutions by $H$, the QFT maps the time register to the quasi-frequency register, and the controlled rotation on $\rqq$ weights the amplitudes by $\sqrt{\gamma(\omega)}$. Conditioned on successful postselection of $\rB$ and $\rqq$, the registers $\rA$ and $\rE$ store the jump label and quasi-frequency, respectively. We define $\rho_{a,\omega} := \gamma(\omega) \hat{A}_a(\omega) \rho (\hat{A}_a(\omega))^{\dagger}$.} \label{fig:block_encoding_u} 
\end{figure*}

This filtered-jump block encoding is the central ingredient in the weak-measurement simulation of the dissipative part of the Lindbladian discussed in Sec.~II. Standard Gibbs-sampling implementations trace out the resulting ancilla record, thereby recovering the Lindbladian evolution alone. In the main text, by contrast, we retain the quasi-frequency outcomes and use their statistics as a convergence diagnostic.


\paragraph{Distribution of quasi-frequencies.}
\label{app:distribution_quasi_frequencies}

We now describe in more detail the quasi-frequency record produced by the weak-measurement implementation of the dissipative step. Operationally, each iteration produces a collection of flag bits and, when successful, a jump label together with a quasi-frequency.

The registers \(\rB\) and \(\rqq\) flag whether the block-encoding routine \(U\) was successfully applied. Conditioned on \(\rB=\rqq=0\), the register \(\rqq'\) distinguishes between a decay event \((\rqq'=0)\) and a transition event \((\rqq'=1)\). In the latter case, the register \(\rE\) stores the quasi-frequency \(\omega\), and the register \(\rA\) stores the jump label \(a\).

For a small time step \(\delta\), the probability density of observing a successful transition event with jump label \(a\) and quasi-frequency in \([\omega,\omega+d\omega)\) is
\begin{equation}
\label{eq:app_joint_transition_density}
    \delta\,\gamma(\omega)\,
    \Tr\!\left[\hat A_a^\dagger(\omega)\hat A_a(\omega)\rho\right]\,d\omega.
\end{equation}
Equivalently, the corresponding unnormalized post-measurement state in register \(\rS\) is proportional to
\begin{equation}
    \gamma(\omega)\,\hat A_a(\omega)\rho \hat A_a^\dagger(\omega).
\end{equation}
Summing \eqref{eq:app_joint_transition_density} over all jump labels and conditioning on a transition event yields the observed quasi-frequency density
\begin{equation}
\label{eq:app_mu_trans}
\mu_{\mathrm{trans}}(\omega,\rho)
=
\frac{
\sum_{a\in\mathcal A}
\gamma(\omega)\Tr\!\left[\hat A_a^\dagger(\omega)\hat A_a(\omega)\rho\right]
}{
\int d\omega'
\sum_{a\in\mathcal A}
\gamma(\omega')\Tr\!\left[\hat A_a^\dagger(\omega')\hat A_a(\omega')\rho\right]
}.
\end{equation}
At equilibrium we denote
\begin{equation}
    \pi(\omega):=\mu_{\mathrm{trans}}(\omega,\rho_\beta).
\end{equation}

For the proof of Proposition~\ref{thm:quasi_frequency_symmetry}, it is convenient to work with the unnormalized equilibrium density
\begin{equation}
\label{eq:app_h_definition}
    h(\omega):=
    \gamma(\omega)\sum_{a\in\mathcal A}
    \Tr\!\left[\hat A_a^\dagger(\omega)\hat A_a(\omega)\rho_\beta\right],
\end{equation}
so that
\begin{equation}
    \pi(\omega)=\frac{h(\omega)}{\int h(\omega')\,d\omega'}.
\end{equation}

Using the Bohr-frequency decomposition
\begin{equation}
    \hat A_a(\omega)=\sum_{\nu\in\mathcal B(H)}\hat A_a(\nu)\,\hat f(\omega-\nu),
\end{equation}
we obtain
\begin{equation}
\begin{split}
    h(\omega)
    &=
    \gamma(\omega)
    \sum_{a\in\mathcal A}
    \sum_{\nu_1,\nu_2\in\mathcal B(H)}
    \overline{\hat f(\omega-\nu_1)}\,\hat f(\omega-\nu_2)\, 
    \\
    &\qquad \times \Tr\!\left[\hat A_a^\dagger(\nu_1)\hat A_a(\nu_2)\rho_\beta\right].
\end{split}
\end{equation}
The off-diagonal terms vanish. Indeed, from
\[
[H,\hat A_a^\dagger(\nu_1)\hat A_a(\nu_2)]
=
(\nu_2-\nu_1)\hat A_a^\dagger(\nu_1)\hat A_a(\nu_2)
\]
and \([H,\rho_\beta]=0\), one finds that
\[
\Tr\!\left[\hat A_a^\dagger(\nu_1)\hat A_a(\nu_2)\rho_\beta\right]=0
\qquad
\text{whenever }\nu_1\neq \nu_2.
\]
Therefore
\begin{equation}
\label{eq:app_h_diagonal_form}
    h(\omega)
    =
    \gamma(\omega)\sum_{\nu\in\mathcal B(H)}
    |\hat f(\omega-\nu)|^2\,m(\nu),
\end{equation}
where we defined
\begin{equation}
\label{eq:app_m_nu_definition}
    m(\nu):=
    \sum_{a\in\mathcal A}
    \Tr\!\left[\hat A_a(\nu)\rho_\beta \hat A_a^\dagger(\nu)\right].
\end{equation}

When the filter is Gaussian, \(|\hat f(\omega-\nu)|^2\) is, up to an irrelevant positive normalization constant, the Gaussian profile
\begin{equation}
\label{eq:app_phi_sigmaE}
    \phi_{\sigma_E}(\omega-\nu):=
    \exp\!\left(-\frac{(\omega-\nu)^2}{2\sigma_E^2}\right).
\end{equation}
Thus the equilibrium quasi-frequency distribution is obtained by Gaussian broadening of the discrete Davies weights \(m(\nu)\), followed by reweighting with the acceptance factor \(\gamma(\omega)\).

In the Davies limit \(\sigma_E\to0\), the Gaussian profile sharpens to a delta peak at each Bohr frequency. One then recovers the discrete equilibrium distribution over exact frequencies, whose symmetry around \(0\) is the content of Lemma~\ref{lem:Davies_Mu}. Proposition~\ref{thm:quasi_frequency_symmetry} shows that for finite Gaussian resolution this symmetry persists in a shifted form around the center
\begin{equation}
    C=-\frac{\beta\sigma_E^2}{2}.
\end{equation}
The following appendix provides more details.


\section{Proof of Proposition~\ref{thm:quasi_frequency_symmetry}}
\label{app:proof_quasi_frequency_symmetry}

We start with some useful properties.

\begin{lemma}
\label{lem:DB_davies_identity}
Define
\begin{equation}
    m(\nu):=
    \sum_{a\in\mathcal A}
    \Tr\!\left[\hat A_a(\nu)\rho_\beta \hat A_a^\dagger(\nu)\right].
\end{equation}
Then
\begin{equation}
    m(\nu)=e^{\beta\nu}m(-\nu).
\end{equation}
\end{lemma}

\begin{proof}
Follows immediately from Lemma  \ref{lem:Davies_Mu} and Eq. \eqref{eq:KMS_condition_first_davies}.

\end{proof}

\begin{lemma}
\label{lem:Gaussian_reflection}
For the unnormalized Gaussian
\begin{equation}
    \phi_\sigma(x):=\exp\!\left(-\frac{x^2}{2\sigma^2}\right),
\end{equation}
one has
\begin{equation}
    \phi_\sigma(x-a)
    =
    \phi_\sigma(x+a)\,
    \exp\!\left(\frac{2ax}{\sigma^2}\right).
\end{equation}
\end{lemma}

\begin{proof}
Expanding the squares gives
\begin{equation}
\frac{\phi_\sigma(x-a)}{\phi_\sigma(x+a)}
=
\exp\!\left(
-\frac{(x-a)^2-(x+a)^2}{2\sigma^2}
\right)
=
\exp\!\left(\frac{2ax}{\sigma^2}\right).
\qedhere
\end{equation}

\end{proof}

\begin{proof}[Proof of Proposition~\ref{thm:quasi_frequency_symmetry}]
Since \(|\hat f(\omega-\nu)|^2\) is, up to an overall positive constant, equal to \(\phi_{\sigma_E}(\omega-\nu)\), it suffices to work with
\begin{equation}
\label{eq:app_h_phi_form}
    h(\omega)
    =
    \gamma(\omega)\sum_{\nu\in\mathcal B(H)}
    \phi_{\sigma_E}(\omega-\nu)\,m(\nu).
\end{equation}
Let
\begin{equation}
C:=-\frac{\beta\sigma_E^2}{2}.
\end{equation}
We show that $h(C-u)=h(C+u)$ for all $u\in\mathbb R$.

Starting from \eqref{eq:app_h_phi_form},
\begin{align}
    h(C-u)
    &=
    \gamma(C-u)\sum_{\nu\in\mathcal B(H)}
    \phi_{\sigma_E}(C-u-\nu)\,m(\nu)
    \nonumber\\
    &=
    \gamma(C-u)\sum_{\nu\in\mathcal B(H)}
    \phi_{\sigma_E}(C+u+\nu)\,
    \nonumber\\
    &\qquad\times
    \exp\!\left(\frac{2C(u+\nu)}{\sigma_E^2}\right)\,
    m(\nu).
\end{align}
where we used Lemma~\ref{lem:Gaussian_reflection} with $x=C$ and $a=u+\nu$. Using Lemma~\ref{lem:DB_davies_identity},
\begin{align}
\begin{split}
    h(C-u)
    &=
    \gamma(C-u)\sum_{\nu\in\mathcal B(H)}
    \phi_{\sigma_E}(C+u+\nu)\, 
    \\ &\qquad \times
    \exp\!\left(\frac{2C(u+\nu)}{\sigma_E^2}\right)\,
    e^{\beta\nu}m(-\nu).
\end{split}
\end{align}
Relabeling $\nu\mapsto -\nu$ gives
\begin{align}
\begin{split}
    h(C-u)
    &=
    \gamma(C-u)\sum_{\nu\in\mathcal B(H)}
    \phi_{\sigma_E}(C+u-\nu)\,
    \\ &\qquad \times
    \exp\!\left(\frac{2C(u-\nu)}{\sigma_E^2}\right)\,
    e^{-\beta\nu}m(\nu).
\end{split}
\end{align}
Now using $C=-\beta\sigma_E^2/2$, the exponential factors collapse:
\begin{equation}
\exp\!\left(\frac{2C(u-\nu)}{\sigma_E^2}\right)e^{-\beta\nu}
=
e^{-\beta u}.
\end{equation}
Hence
\begin{align}
    h(C-u)
    &=
    \gamma(C-u)e^{-\beta u}
    \sum_{\nu\in\mathcal B(H)}
    \phi_{\sigma_E}(C+u-\nu)\,m(\nu).
\end{align}
Applying the reflection identity \eqref{eq:rate_reflection_identity},
\begin{equation}
\gamma(C-u)e^{-\beta u}=\gamma(C+u),
\end{equation}
we conclude that
\begin{align*}
h(C-u)&=\gamma(C+u)\sum_{\nu\in\mathcal B(H)}
\phi_{\sigma_E}(C+u-\nu)\,m(\nu)\\
&=h(C+u).
\end{align*}
Therefore \(\pi(C-u)=\pi(C+u)\), and the statement \(\langle \omega\rangle_\pi=C\) follows immediately.

\end{proof}

\begin{remark}
The reflection identity \eqref{eq:rate_reflection_identity} is linear in $\gamma$. Hence, if
\begin{equation}
\gamma(\omega)=\sum_j c_j\,\gamma_j(\omega)
\end{equation}
and each component satisfies
\begin{equation}
\gamma_j(C-u)=e^{\beta u}\gamma_j(C+u),
\end{equation}
then the same identity holds for $\gamma$:
\begin{equation}
\gamma(C-u)=e^{\beta u}\gamma(C+u).
\end{equation}

For a single Gaussian component
\begin{equation}
\gamma_j(\omega)=\exp\!\left(-\frac{(\omega+\omega_{\gamma,j})^2}{2\sigma_{\gamma,j}^2}\right),
\end{equation}
the identity holds if and only if
\begin{equation}
\label{eq:app_single_gaussian_beta_relation}
    \beta=\frac{2\omega_{\gamma,j}}{\sigma_{\gamma,j}^2+\sigma_E^2},
\end{equation}
equivalently,
\begin{equation}
\omega_{\gamma,j}=\frac{\beta}{2}\bigl(\sigma_{\gamma,j}^2+\sigma_E^2\bigr).
\end{equation}
Thus a discrete or continuous mixture of Gaussians works provided all components satisfy \eqref{eq:app_single_gaussian_beta_relation} for the same global $\beta$ and the same filter width $\sigma_E$.

In particular, this covers the Gaussian-mixture family of Chen, Kastoryano, and Gily\'en \cite{Chen2023a}, who fix $\sigma_E$ and define
\begin{equation}
\gamma^{(g)}(\omega)
=
\int_{\beta\sigma_E^2/2}^{\infty}
g(x)\,
\exp\!\left(
-\frac{(\omega+x)^2}{4x/\beta-2\sigma_E^2}
\right)\,dx.
\end{equation}
Each integrand is a Gaussian with
\begin{equation}
\omega_\gamma(x)=x,
\qquad
\sigma_\gamma(x)^2=\frac{2x}{\beta}-\sigma_E^2,
\end{equation}
and therefore satisfies
\begin{equation}
\beta=\frac{2\omega_\gamma(x)}{\sigma_\gamma(x)^2+\sigma_E^2}.
\end{equation}
Hence every component has the same reflection center
\begin{equation}
C=-\frac{\beta\sigma_E^2}{2},
\end{equation}
and the full mixture obeys \eqref{eq:rate_reflection_identity} by linearity. Consequently, Proposition~\ref{thm:quasi_frequency_symmetry} applies to the entire linear-combination family of Corollary II.3 in \cite{Chen2023a}. In particular, it also covers the special choices of Proposition II.4, including the shifted Metropolis filter and the smooth Glauber-like filter.

For the Metropolis choice
\begin{equation}
\gamma(\omega)=\exp\!\left(-\beta \max\!\left(\omega+\frac{\beta\sigma_E^2}{2},0\right)\right),
\end{equation}
the identity follows by a direct case split on the sign of $u$, using $C=-\beta\sigma_E^2/2$.
\end{remark}

\begin{proposition}
\label{prop:universal_skew_implies_reflection}
Let
\begin{equation}
\hat f(\omega)=A_E\exp\!\left(-\frac{\omega^2}{4\sigma_E^2}\right)
\end{equation}
be a Gaussian filter, and let
\begin{equation}
\alpha_{\nu_1,\nu_2}
:=
\int\gamma(\omega)\,\hat f(\omega-\nu_1)\hat f(\omega-\nu_2)\,d\omega.
\end{equation}
Assume that $\gamma$ is a finite Gaussian mixture and that
\begin{equation}
\label{eq:universal_skew_symmetry}
\alpha_{\nu_1,\nu_2}
=
e^{-\beta(\nu_1+\nu_2)/2}\alpha_{-\nu_2,-\nu_1}
\qquad
\text{for all }\nu_1,\nu_2\in\mathbb R.
\end{equation}
Then, setting
\begin{equation}
C:=-\frac{\beta\sigma_E^2}{2},
\end{equation}
the rate $\gamma$ satisfies
\begin{equation}
\label{eq:reflection_ansatz_again}
\gamma(C-u)=e^{\beta u}\gamma(C+u)
\qquad
\text{for all }u\in\mathbb R.
\end{equation}
\end{proposition}

\begin{proof}
Write
\begin{equation}
s:=\frac{\nu_1+\nu_2}{2}.
\end{equation}
A direct Gaussian multiplication gives
\begin{multline}
\hat f(\omega-\nu_1)\hat f(\omega-\nu_2) \\
{}= K_E\,
\exp\!\Bigl(
-\tfrac{(\nu_1-\nu_2)^2 + 4(\omega-s)^2}{8\sigma_E^2}
\Bigr).
\end{multline}
where $K_E>0$ is independent of $\nu_1,\nu_2,\omega$. Hence,
\begin{equation}
\alpha_{\nu_1,\nu_2}
=
K_E\,
\exp\!\left(-\frac{(\nu_1-\nu_2)^2}{8\sigma_E^2}\right)\,
G(s),
\end{equation}
with
\begin{equation}
G(s):=(\gamma * \phi_{\sigma_E})(s),
\qquad
\phi_{\sigma_E}(x):=\exp\!\left(-\frac{x^2}{2\sigma_E^2}\right).
\end{equation}
Therefore \eqref{eq:universal_skew_symmetry} is equivalent to
\begin{equation}
G(s)=e^{-\beta s}G(-s)
\qquad
\text{for all }s\in\mathbb R.
\end{equation}

Define
\begin{equation}
\widetilde\gamma(\omega):=e^{\beta\omega/2}\gamma(\omega).
\end{equation}
By completing the square in $\phi _{\sigma _E}$ 
\begin{align*}
e^{\beta s/2}G(s)&=e^{\beta s/2}\int\gamma(\omega)\phi_{\sigma_E}(s-\omega)\,d\omega,  \\
&=\int \widetilde\gamma(\omega) \exp\!\left( \frac{\beta}{2}(s-\omega)-\frac{(s-\omega)^2 }{2\sigma_E^2} \right) \,d\omega, \\
&=\int \widetilde\gamma(\omega) \exp\!\left( -\frac{((s+C)-\omega)^2}{2\sigma_E^2}
+\frac{\beta^2\sigma_E^2}{8}\right)\,d\omega,
\end{align*}
with $C=-\frac{\beta\sigma_E^2}{2}$. Therefore
\begin{equation}
e^{\beta s/2}G(s)
=
e^{\beta^2\sigma_E^2/8}\,
(\widetilde\gamma * \phi_{\sigma_E})(s+C).
\end{equation}
Thus, the relation $G(s)=e^{-\beta s}G(-s)$ is equivalent to
\begin{equation}
(\widetilde\gamma * \phi_{\sigma_E})(C+u)
=
(\widetilde\gamma * \phi_{\sigma_E})(C-u)
\qquad
\text{for all }u\in\mathbb R.
\end{equation}
Now, let
\begin{equation}
\psi(u):=\widetilde\gamma(C+u)-\widetilde\gamma(C-u).
\end{equation}
Then $\psi * \phi_{\sigma_E}=0$. Taking Fourier transforms,
\begin{equation}
\widehat{\psi}(\xi)\,\widehat{\phi}_{\sigma_E}(\xi)=0.
\end{equation}
Since
\begin{equation}
\widehat{\phi}_{\sigma_E}(\xi)\propto e^{-\sigma_E^2\xi^2/2}
\end{equation}
never vanishes, we must have $\widehat{\psi}=0$, hence $\psi=0$. Therefore
\begin{equation}
\widetilde\gamma(C+u)=\widetilde\gamma(C-u)
\qquad
\text{for all }u\in\mathbb R.
\end{equation}
Undoing the definition of $\widetilde\gamma$ gives
\begin{equation}
e^{\beta(C+u)/2}\gamma(C+u)=e^{\beta(C-u)/2}\gamma(C-u),
\end{equation}
which is exactly
\begin{equation}
\gamma(C-u)=e^{\beta u}\gamma(C+u).
\end{equation}

\end{proof}

\begin{remark}
For a Gaussian filter, the reflection identity \eqref{eq:rate_reflection_identity} is not limited to the aligned Gaussian family explicitly constructed in \cite{Chen2023a}. Proposition \ref{prop:universal_skew_implies_reflection} shows that any finite Gaussian-mixture rate satisfying the skew symmetry \eqref{eq:universal_skew_symmetry} must satisfy \eqref{eq:rate_reflection_identity}. The general solution in that case, is obtained by combining aligned singletons
\begin{equation}
\omega=\frac{\beta}{2}\bigl(\sigma^2+\sigma_E^2\bigr)
\end{equation}
and same-variance reflected pairs
\begin{equation}
c_-\,e^{-(\omega+\omega_-)^2/(2\sigma^2)}
+
c_+\,e^{-(\omega+\omega_+)^2/(2\sigma^2)},
\end{equation}
with
\begin{equation}
\omega_+ + \omega_- = \beta(\sigma^2+\sigma_E^2),
\qquad
c_+ = c_-\,e^{\beta(\omega_+ - \omega_-)/2}.
\end{equation}
The aligned Gaussian family of \cite{Chen2023a} is the special case where every component is a singleton.
\end{remark}


\section{Qubit efficient thermalization and monitoring}
\label{app:qubit_efficient_monitoring}

In this appendix we justify the claim made in Sec. \ref{sec:qgs} that the one-ancilla algorithm of Ref.~\cite{Ding2025} admits a natural monitored frequency record and, for the Gaussian prior of Appendix~I therein, reduces to the same distributional framework analyzed in Sec. \ref{sec:convergence}.

\paragraph{One-ancilla thermalization step.}
Fix a system Hamiltonian $H$ and a family of coupling operators $\mathcal A=\{A_i,-A_i\}_i$ satisfying $\{A_i^\dagger\}_i=\{A_i\}_i$. In one round of the qubit-efficient protocol, one samples $A_S\in\mathcal A$ and a bath frequency $\Omega$ from a prior distribution $g$, prepares a one-qubit bath in
\begin{equation}
\rho_E(\Omega)=\frac{e^{-\beta H_E(\Omega)}}{\Tr(e^{-\beta H_E(\Omega)})},
\qquad
H_E(\Omega)=-\frac{\Omega}{2}Z,
\end{equation}
and evolves under the time-dependent Hamiltonian
\begin{align}
H_\alpha(t)
&=
H+H_E(\Omega)
+
\alpha f(t)\bigl(A_S\otimes B_E+A_S^\dagger\otimes B_E^\dagger\bigr), \nonumber
\\
B_E&=|1\rangle\!\langle 0|.
\end{align}
Here $f$ is a real, even, normalized envelope, taken in \cite{Ding2025} to be Gaussian,
\begin{equation}
f(t)=\frac{1}{(2\pi)^{1/4}\sigma^{1/2}}e^{-t^2/(4\sigma^2)}.
\end{equation}
Averaging over the sampled classical data yields the channel
\begin{equation*}
\Phi_\alpha(\rho)
=
\mathbb E_{A_S,\Omega}\,
\Tr_E\!\Bigl[
U_{\alpha,A_S,\Omega}(T)
\bigl(\rho\otimes \rho_E(\Omega)\bigr)
U_{\alpha,A_S,\Omega}(T)^\dagger
\Bigr].
\end{equation*}
This is the qubit-efficient thermalization primitive of Ref.~\cite{Ding2025}.

\paragraph{Monitored signed record.}
Although the algorithm normally discards the ancilla outcomes, the same circuit naturally provides a weak-measurement record. Since $\rho_E(\Omega)$ is diagonal in the computational basis, one may conceptually resolve each round into the two initial bath branches $|0\rangle$ and $|1\rangle$. We retain only rounds in which the bath qubit flips and define the signed accepted frequency
\begin{equation}
\omega=
\begin{cases}
+\Omega,& |1\rangle\to |0\rangle,\\[0.5ex]
-\Omega,& |0\rangle\to |1\rangle.
\end{cases}
\end{equation}
Rounds with no ancilla flip are discarded. This is the direct analogue of the transition-conditioned quasi-frequency record used in Sec. \ref{sec:convergence}. The sign is essential: it distinguishes absorption from emission and therefore retains the detailed-balance information that is lost in the unsigned histogram.

\paragraph{Effective frequency-resolved Lindbladian.}
The weak coupling expansion in Sec. III of \cite{Ding2025} shows that, up to the surrounding free system evolution, one application of $\Phi_\alpha$ implements an $\alpha^2-$time step of an effective Lindbladian with a remainder of order $\alpha^4T^4\|f\|^4_{L^\infty}$. Writing
\begin{equation}
V_{A_S,f,T}(\omega)
:=
\int_{-T}^{T}
f(t)\,e^{itH}A_S e^{-itH}e^{-i\omega t}\,dt,
\end{equation}
the dissipative part of the effective generator contains the two jump branches
\begin{equation*}
\frac{1}{1+e^{\beta\Omega}}\,
\mathcal D_{V_{A_S^\dagger,f,T}(\Omega)}
\qquad\text{and}\qquad
\frac{1}{1+e^{-\beta\Omega}}\,
\mathcal D_{V_{A_S,f,T}(-\Omega)}.
\end{equation*}
Accordingly, once the ancilla-flip outcome is converted into the signed variable $\omega\in\mathbb R$, the monitored accepted-event statistic is governed by a frequency-resolved jump family of exactly the same type as in Sec. \ref{sec:qgs}.

\paragraph{Gaussian prior and Gaussian KMS surrogate.}
For thermal state preparation, Appendix~I of Ref.~\cite{Ding2025} introduces the Gaussian prior
\begin{equation*}
g_x(\omega)
=
\frac{1}{Z_x}
\exp\!\left(
-\frac{(\omega+x)^2}{2\left(2x/\beta-\frac{1}{4\sigma^2}\right)}
\right),
\qquad
x>\frac{\beta}{8\sigma^2},
\end{equation*}
and proves that the native generator is approximated by a Gaussian KMS surrogate $\hat{\mathcal{L}}_{{\rm{KMS}},x}$ of the form
\begin{equation}
\mathbb E_{A_S}\!\left(
-i\left[\frac{B_{A_S}}{Z_x},\rho\right]
+
\int_{\mathbb R}
g_x(\omega)\,
\mathcal D_{V_{A_S,f,\infty}(\omega)}(\rho)\,d\omega
\right).
\end{equation}
up to induced $1$ error of
\begin{equation*}
O\!\left(
\sigma e^{-T^2/(4\sigma^2)}
+\frac{\alpha^2T^4}{\sigma^2}
+\frac{1}{Z_x}\left[
\frac{\beta^2}{\sigma^2}
\frac{x}{x-\beta/(8\sigma^2)}
+\frac{\beta}{\sigma}
\right]
\right),
\end{equation*}
as shown in Appendix I. Thus, the surrogate captures the weak-coupling generator in the large-$T$, large-$\sigma$ regime.

\paragraph{Identification with the transition-conditioned framework of Sec. \ref{sec:convergence}.}
The surrogate above is exactly of the Gaussian type studied in Sec. \ref{sec:convergence}. Indeed, the Gaussian time envelope gives the frequency resolution
\begin{equation}
\sigma_E=\frac{1}{2\sigma},
\end{equation}
and Appendix I of Ref.~\cite{Ding2025} identifies the Gaussian prior with a Gaussian rate
\begin{equation}
\gamma(\omega)
=
\exp\!\left(
-\frac{(\omega+\omega_\gamma)^2}{2\sigma_\gamma^2}
\right)
\end{equation}
with parameters
\begin{equation}
\omega_\gamma=x,
\qquad
\sigma_\gamma^2=\frac{2x}{\beta}-\sigma_E^2.
\end{equation}
Hence
\begin{equation}
\beta=\frac{2\omega_\gamma}{\sigma_\gamma^2+\sigma_E^2},
\end{equation}
which is exactly the Gaussian detailed-balance relation assumed in Proposition~\ref{thm:quasi_frequency_symmetry}.

Therefore, for the surrogate generator, the signed accepted-frequency density
\begin{equation*}
\mu_{\rm flip}(\omega,\rho)
=
\frac{
\sum_{a\in\mathcal A}
g_x(\omega)\,
\Tr\!\bigl[
V_{A_a,f,\infty}^\dagger(\omega)V_{A_a,f,\infty}(\omega)\rho
\bigr]
}{
\int_{\mathbb R} d\omega'
\sum_{a\in\mathcal A}
g_x(\omega')\,
\Tr\!\bigl[
V_{A_a,f,\infty}^\dagger(\omega')V_{A_a,f,\infty}(\omega')\rho
\bigr]
}
\end{equation*}
is of the same form as $\mu_{\mathrm{trans}}$ in Sec. \ref{sec:convergence}. At equilibrium,
\begin{equation*}
\pi_{\rm flip}(\omega):=\mu_{\rm flip}(\omega,\rho_\beta),
\end{equation*}
so the equilibrium symmetry results apply directly. In particular, once Proposition~\ref{thm:quasi_frequency_symmetry} has been established, one obtains the shifted reflection symmetry
\begin{equation}
\begin{aligned}
\pi_{\rm flip}(C-u)&=\pi_{\rm flip}(C+u),
\\
\langle \omega\rangle_{\pi_{\rm flip}}&=C,
\\
C=-\frac{\beta\sigma_E^2}{2}&=-\frac{\beta}{8\sigma^2}.
\end{aligned}
\end{equation}
Thus the monitored ancilla-flip record of the Gaussian-prior qubit-efficient construction is governed by the same equilibrium symmetry principle as the quasi-frequency distribution studied in the main text. For the native one-ancilla channel, this should be understood as the corresponding asymptotic statement in the regime where the error terms displayed above are small.


\section{Technical aspects of the stopping rules}
\label{app:stopping_rule_details}

This appendix states the stochastic output analysis assumptions, defines the covariance estimators and confidence regions, and gives the exact inclusion conditions used by the coordinatewise and ellipsoidal stopping rules.

\paragraph{Accepted-event output process.}
Let $\omega_i$ be the quasi-frequency recorded at the $i$-th accepted transition and define
\begin{equation}
\label{eq:app_Y_definition}
    Y_i
    =
    \begin{pmatrix}
        \omega_i\\
        (\omega_i-C)^3
    \end{pmatrix}
    \in\mathbb R^2.
\end{equation}
The complete weak-measurement record also contains a jump label. A direct indicator or binned representation of that full record is typically high-dimensional and subject to deterministic linear constraints, which make the covariance of the limit distribution singular. Rarely occupied components can also make its empirical covariance poorly conditioned. For this reason, we use two-dimensional statistic in Eq.~\eqref{eq:app_Y_definition}, which instead targets the equilibrium constraints of Proposition \ref{thm:quasi_frequency_symmetry}, for which
\begin{equation}
\label{eq:app_theta_target}
    \theta^\star
    :=
    \mathbb E_\pi[Y_i]
    =
    \begin{pmatrix}
        C\\
        0
    \end{pmatrix},
\end{equation}
where the target $\theta^\star$ is fixed by the theoretical symmetry relation.

The output analysis is conditional a multivariate central limit
theorem for the stationary accepted-event process:
\begin{equation}
\label{eq:app_ergodic_clt}
\begin{aligned}
    \widehat\theta_n
    &:={}
    \frac{1}{n}\sum_{i=1}^{n}Y_i
    \xrightarrow{\mathrm{a.s.}}
    \theta^\star,\\
    \sqrt n\bigl(\widehat\theta_n-\theta^\star\bigr)
    &\xrightarrow{d}
    \mathcal N(0,\Sigma).
\end{aligned}
\end{equation}
which is supplied by \cite{attal2015central} (see \cite{bringuier2017central} for a continuous version). Here $\Sigma\succ0$ is the long-run covariance matrix of the monitored process. The accepted-event output need not itself be represented as a first-order Markov chain; the properties in Eq.~\eqref{eq:app_ergodic_clt}, together with consistency of the covariance estimator below, are what the stopping analysis requires. In the finite-register implementation, $\omega_i$ lies in a bounded interval and hence $Y_i$ has moments of all orders. Boundedness supplies the required moment conditions, but the ergodic and mixing conditions needed for the CLT remain separate assumptions.

It is important to distinguish the long-run covariance $\Sigma$ from the stationary marginal covariance
\begin{equation}
\label{eq:app_sigma_lambda_distinction}
\begin{aligned}
    \Lambda
    &:={}
    \operatorname{Var}_\pi(Y_i),\\
    \Sigma
    &=
    \sum_{\ell=-\infty}^{\infty}
    \operatorname{Cov}_\pi(Y_0,Y_\ell),
\end{aligned}
\end{equation}
whenever the covariance series is well defined. The matrix $\Sigma$ determines the Monte Carlo uncertainty of the sample mean, whereas $\Lambda$ provides a marginal scale and geometry for comparing the two monitored coordinates.

\paragraph{Selected windows and non-overlapping batch means.}
Suppose that $N$ accepted events have been observed. We first discard an absolute accepted-event burn-in $N_0$. The implementation may then retain only a late fraction of the post-burn-in record. For a batch-means discard fraction $f_{\mathrm{BM}}\in[0,1)$, define
\begin{equation}
\label{eq:app_selected_window}
\begin{aligned}
    r_N^{\mathrm{BM}}
    &=
    N_0+
    \left\lfloor
        f_{\mathrm{BM}}(N-N_0)
    \right\rfloor,\\
    L_N^{\mathrm{BM}}
    &=
    N-r_N^{\mathrm{BM}}.
\end{aligned}
\end{equation}
Thus $f_{\mathrm{BM}}=0$ uses the entire post-burn-in accepted record, while $f_{\mathrm{BM}}>0$ restricts the calculation to a later accepted-event window. The pseudocode in the main text shows the basic choice $f_{\mathrm{BM}}=0$; the numerical implementation allows the more general choice in Eq.~\eqref{eq:app_selected_window}.

Within the selected window, choose
\begin{equation}
\label{eq:app_batch_sizes}
\begin{aligned}
    a
    &=
    \left\lfloor\sqrt{L_N^{\mathrm{BM}}}\right\rfloor,
    &
    b
    &=
    \left\lfloor\frac{L_N^{\mathrm{BM}}}{a}\right\rfloor,\\
    M
    &=ab,
    &
    s
    &=N-M,
\end{aligned}
\end{equation}
and use the last $M$ observations of the selected window. The corresponding mean and non-overlapping batch means are
\begin{equation}
\label{eq:app_theta_and_batch_means}
\begin{aligned}
    \widehat\theta_M
    &=
    \frac{1}{M}
    \sum_{i=s+1}^{N}Y_i,\\
    \overline Y_k
    &=
    \frac{1}{b}
    \sum_{j=1}^{b}Y_{s+kb+j},
    \qquad k=0,\ldots,a-1.
\end{aligned}
\end{equation}
The non-overlapping batch-means estimator of the long-run covariance is
\begin{equation}
\label{eq:app_batch_means_covariance}
    \widehat\Sigma_{\mathrm{BM},M}
    =
    \frac{b}{a-1}
    \sum_{k=0}^{a-1}
    (\overline Y_k-\widehat\theta_M)
    (\overline Y_k-\widehat\theta_M)^\top.
\end{equation}
Under the usual mixing, moment, and batch-growth conditions, $\widehat\Sigma_{\mathrm{BM},M}\xrightarrow{\mathrm{a.s.}}\Sigma$ \cite{flegal2010batch,vats2019multivariate}. In particular, the standard asymptotic regime requires $a\to\infty$, $b\to\infty$, and $b/M\to0$. The square-root choice in Eq.~\eqref{eq:app_batch_sizes} satisfies these basic growth relations. If a moving late window is used, the analysis additionally assumes that the selected tail obeys the same
CLT and batch-means consistency conditions.

\paragraph{Checkpointwise confidence regions.}
Treating the batch means as approximately independent Gaussian vectors gives the Hotelling-type finite-batch critical value
\begin{equation}
\label{eq:app_hotelling_critical}
    c_{\alpha,a}
    =
    \frac{p(a-1)}{a-p}
    F_{1-\alpha;\,p,\,a-p},
    \qquad a>p,
\end{equation}
where $p=2$ in the present application.  The corresponding confidence ellipsoid is
\begin{equation}
\label{eq:app_confidence_ellipsoid}
    \mathcal C_\alpha(M)
    =
    \{\theta\in\mathbb R^p:  M(\widehat\theta_M-\theta)^\top
    \widehat\Sigma_{\mathrm{BM},M}^{-1}
    (\widehat\theta_M-\theta) \le c_{\alpha,a}\}.
\end{equation}
Equivalently, with
\begin{equation}
\label{eq:app_A_definition}
    A_M
    =
    \frac{c_{\alpha,a}}{M}
    \widehat\Sigma_{\mathrm{BM},M},
\end{equation}
one may write
\begin{equation}
\label{eq:app_confidence_ellipsoid_affine}
\begin{aligned}
    \mathcal C_\alpha(M)
    &=
    \widehat\theta_M+A_M^{1/2}\mathbb B_p,\\
    \mathbb B_p
    &=
    \{u\in\mathbb R^p:\|u\|_2\le1\}.
\end{aligned}
\end{equation}
At a deterministic checkpoint, Eqs.~\eqref{eq:app_confidence_ellipsoid}--
\eqref{eq:app_confidence_ellipsoid_affine} have the usual asymptotic
$1-\alpha$ confidence interpretation under the CLT and covariance-consistency assumptions.
The finite-batch $F$ correction is exact only under the idealized independent Gaussian
batch-means model and is otherwise an output-analysis approximation.

\paragraph{Relation to relative fixed-volume rules and multivariate ESS.}
The volume of the ellipsoid in Eq.~\eqref{eq:app_confidence_ellipsoid} is
\begin{equation}
\label{eq:app_confidence_volume}
    \operatorname{Vol}(\mathcal C_\alpha(M))
    =
    \frac{\pi^{p/2}}{\Gamma(p/2+1)}
    \left(\frac{c_{\alpha,a}}{M}\right)^{p/2}
    \left|\widehat\Sigma_{\mathrm{BM},M}\right|^{1/2}.
\end{equation}
For comparison with Ref.~\cite{vats2019multivariate}, let
\begin{equation}
\label{eq:app_lambda_window_estimator}
    \widehat\Lambda_M^{\mathrm{win}}
    =
    \frac{1}{M-1}
    \sum_{i=s+1}^{N}
    (Y_i-\widehat\theta_M)
    (Y_i-\widehat\theta_M)^\top
\end{equation}
be the ordinary sample covariance on the same window. Omitting the usual minimum-sample
safeguard, the relative fixed-volume rule monitors
\begin{equation}
\label{eq:app_relative_volume}
    v_M
    =
    \frac{
        \operatorname{Vol}(\mathcal C_\alpha(M))^{1/p}
    }{
        |\widehat\Lambda_M^{\mathrm{win}}|^{1/(2p)}
    },
\end{equation}
and stops when $v_M$ is below a prescribed relative-precision threshold. The associated
multivariate effective sample size is
\begin{equation}
\label{eq:app_mess}
    \widehat{\operatorname{mESS}}_M
    =
    M
    \left(
        \frac{|\widehat\Lambda_M^{\mathrm{win}}|}
             {|\widehat\Sigma_{\mathrm{BM},M}|}
    \right)^{1/p}.
\end{equation}
For $p=2$ these quantities satisfy
\begin{equation}
\label{eq:app_volume_mess_relation}
    v_M^2
    =
    \frac{\pi c_{\alpha,a}}
         {\widehat{\operatorname{mESS}}_M}.
\end{equation}
Thus relative volume and multivariate ESS characterize the size of the confidence region
relative to the marginal scale of $Y_i$. This can also be used as a precision diagnostic.


\paragraph{$\Lambda$-scaled ellipsoid inclusion.}
The ellipsoidal rule uses $\Lambda$ to place the two monitored coordinates on a common marginal scale. The long-run covariance estimate $\widehat\Sigma_{\mathrm{BM},M}$ determines the size and orientation of the confidence ellipsoid, while $\Lambda$ determines the metric of the target tolerance region.

For the fixed-geometry implementation, let $Y_1^{\mathrm{cal}},\ldots,Y_{L_{\mathrm{cal}}}^{\mathrm{cal}}$ be accepted-event observations from an independent stationary calibration run and define
\begin{equation}
\label{eq:app_lambda_calibration}
\begin{aligned}
    \overline Y_{\mathrm{cal}}
    &=
    \frac{1}{L_{\mathrm{cal}}}
    \sum_{i=1}^{L_{\mathrm{cal}}}Y_i^{\mathrm{cal}},\\
    \Lambda_{\mathrm{cal}}
    &=
    \frac{1}{L_{\mathrm{cal}}-1}
    \sum_{i=1}^{L_{\mathrm{cal}}}
    (Y_i^{\mathrm{cal}}-\overline Y_{\mathrm{cal}})
    (Y_i^{\mathrm{cal}}-\overline Y_{\mathrm{cal}})^\top.
\end{aligned}
\end{equation}
The matrix $\Lambda_{\mathrm{cal}}$ is then held fixed during the production trajectory. The nominal level $\alpha$ is conditional on this chosen matrix and does not account for finite calibration error.

Alternatively, the online implementation estimates the marginal covariance from a late window of the current accepted-event trajectory. For a separately chosen $f_\Lambda\in[0,1)$, let
\begin{equation}
\label{eq:app_lambda_online_window}
\begin{aligned}
    r_N^\Lambda
    &=
    N_0+
    \left\lfloor
        f_\Lambda(N-N_0)
    \right\rfloor,\\
    L_N^\Lambda
    &=N-r_N^\Lambda,\\
    \overline Y_\Lambda
    &=
    \frac{1}{L_N^\Lambda}
    \sum_{i=r_N^\Lambda+1}^{N}Y_i,
\end{aligned}
\end{equation}
and set
\begin{equation}
\label{eq:app_lambda_online_estimator}
    \widehat\Lambda_{\mathrm{online},N}
    =
    \frac{1}{L_N^\Lambda-1}
    \sum_{i=r_N^\Lambda+1}^{N}
    (Y_i-\overline Y_\Lambda)
    (Y_i-\overline Y_\Lambda)^\top.
\end{equation}
The fractions $f_{\mathrm{BM}}$ and $f_\Lambda$ are intentionally separate: the first selects the data used for $\widehat\theta_M$ and the long-run covariance, while the second selects the data used only for the marginal geometry. At finite $N$, using $\widehat\Lambda_{\mathrm{online},N}$ makes the target region random, and its estimation uncertainty is not included in the nominal confidence level. The asymptotic interpretation below requires $\widehat\Lambda_{\mathrm{online},N}\xrightarrow{\mathrm{a.s.}}\Lambda\succ0$.

For either choice of scale matrix, define
\begin{equation}
\label{eq:app_target_ellipsoid}
    \mathcal T_\Lambda(P_\Lambda)
    =
    \left\{
        \theta\in\mathbb R^2:
        (\theta-\theta^\star)^\top
        \Lambda^{-1}
        (\theta-\theta^\star)
        \le P_\Lambda^2
    \right\}.
\end{equation}
As stated in the main text, the exact rule is
\begin{equation}
\label{eq:app_ruleII_inclusion}
    \mathcal C_\alpha(M)
    \subseteq
    \mathcal T_\Lambda(P_\Lambda).
\end{equation}
Equivalently, define the target-relative outer radius
\begin{equation}
\label{eq:app_outer_radius_set_form}
    R_{\Lambda,M}
    =
    \sup_{\theta\in\mathcal C_\alpha(M)}
    \left[
        (\theta-\theta^\star)^\top
        \Lambda^{-1}
        (\theta-\theta^\star)
    \right]^{1/2}.
\end{equation}
Then Eq.~\eqref{eq:app_ruleII_inclusion} holds if and only if $R_{\Lambda,M}\le P_\Lambda$.

For numerical evaluation, write
\begin{equation}
\label{eq:app_ruleII_affine_quantities}
    d_M
    =
    \widehat\theta_M-\theta^\star,
    \qquad
    A_M
    =
    \frac{c_{\alpha,a}}{M}
    \widehat\Sigma_{\mathrm{BM},M}.
\end{equation}
Using the affine representation in Eq.~\eqref{eq:app_confidence_ellipsoid_affine}, Eq.~\eqref{eq:app_outer_radius_set_form} becomes
\begin{equation}
\label{eq:app_exact_ellipsoid_radius}
    R_{\Lambda,M}
    =
    \sup_{\|u\|_2\le1}
    \left\|d_M+A_M^{1/2}u\right\|_{\Lambda^{-1}},
\end{equation}
where $\|x\|_{\Lambda^{-1}}=(x^\top\Lambda^{-1}x)^{1/2}$. In two dimensions, a maximizer may be sought on the boundary by setting $u(\phi)=(\cos\phi,\sin\phi)^\top$:
\begin{equation}
\label{eq:app_boundary_search}
    R_{\Lambda,M}^2
    =
    \max_{\phi\in[0,2\pi)}
    \left\|d_M+A_M^{1/2}u(\phi)\right\|_{\Lambda^{-1}}^2.
\end{equation}
The implementation evaluates this one-dimensional maximization directly to numerical tolerance and uses the resulting outer radius for the stopping decision.


\paragraph{Coordinatewise Bonferroni box inclusion.}
Let $\alpha_{\mathrm{tot}}$ be the desired checkpointwise simultaneous error level. With $p=2$, the two-sided Bonferroni critical value is
\begin{equation}
\label{eq:app_bonferroni_critical}
    t^\star_{\mathrm{Bonf}}
    =
    t_{1-\alpha_{\mathrm{tot}}/(2p),\,a-1}.
\end{equation}
The coordinatewise half-widths are
\begin{equation}
\label{eq:app_bonferroni_halfwidths}
\begin{aligned}
    h_{1,M}
    &=
    t^\star_{\mathrm{Bonf}}
    \sqrt{
        \frac{(\widehat\Sigma_{\mathrm{BM},M})_{11}}{M}
    },\\
    h_{3,M}
    &=
    t^\star_{\mathrm{Bonf}}
    \sqrt{
        \frac{(\widehat\Sigma_{\mathrm{BM},M})_{22}}{M}
    }.
\end{aligned}
\end{equation}
At a fixed sufficiently large checkpoint, the resulting rectangular region has asymptotic simultaneous coverage at least $1-\alpha_{\mathrm{tot}}$ by the Bonferroni inequality.
Writing
\begin{equation}
\label{eq:app_coordinate_distances}
    D_{1,M}
    =
    |\widehat m_{1,M}-C|,
    \qquad
    D_{3,M}
    =
    |\widehat m_{3,M}|,
\end{equation}
the confidence and tolerance boxes are
\begin{equation}
\label{eq:app_box_definitions}
\begin{aligned}
    \widehat B_\alpha(M)
    ={}&
    [\widehat m_{1,M}-h_{1,M},
     \widehat m_{1,M}+h_{1,M}]
    \\
    &{}\times
    [\widehat m_{3,M}-h_{3,M},
     \widehat m_{3,M}+h_{3,M}],\\
    B(P_1,P_3)
    ={}&
    [C-P_1,C+P_1]
    \times[-P_3,P_3].
\end{aligned}
\end{equation}
Their inclusion relation is equivalent to two scalar inequalities:
\begin{equation}
\label{eq:app_box_inclusion_equivalence}
\begin{aligned}
    \widehat B_\alpha(M)
    \subseteq B(P_1,P_3)
    \quad\Longleftrightarrow\quad
    D_{1,M}+h_{1,M}&\le P_1,\\
    D_{3,M}+h_{3,M}&\le P_3.
\end{aligned}
\end{equation}
The quantities $D_{j,M}$ measure displacement from the theoretical target, whereas $h_{j,M}$ measure residual Monte Carlo uncertainty. Equation~\eqref{eq:app_box_inclusion_equivalence} controls their sum; it does not require the separate condition $D_{j,M}\le h_{j,M}$ used in
the earlier draft rule.



\paragraph{Choice of tolerance parameters.}
The confidence levels and tolerance radii play different roles. The parameters $\alpha_{\mathrm{tot}}$ and $\alpha$ control the checkpointwise Monte Carlo confidence regions. By contrast, $P_1$, $P_3$, and $P_\Lambda$ encode the largest deviations from the known equilibrium constraints that are considered scientifically acceptable. Thus the tolerances are chosen separately from the confidence levels.

The box tolerances have the units
\begin{equation}
    [P_1]=[\omega],
    \qquad
    [P_3]=[\omega]^3,
\end{equation}
whereas $P_\Lambda$ is dimensionless when $\Lambda$ is a covariance matrix of $Y_i$. They may be fixed from an a priori physical accuracy requirement or selected on independent calibration data. In the numerical benchmarks, a practical choice is obtained by comparing the rule distances with an independently computed physical error, such as the Gibbs energy error, and selecting thresholds that achieve the desired empirical operating point. Such a calibration is an empirical design procedure; without an additional analytical inequality, it does not by itself prove that passing a moment-based rule bounds the error of every observable. Calibration and performance assessment should therefore use independent or held-out trajectories whenever possible.

\paragraph{Asymptotic interpretation of the inclusion rules.}
Consider a sequence of admissible checkpoints for which $M\to\infty$ and $a\to\infty$. To separate the stopping-rule logic from the assumption that equilibrium has already been reached, suppose more generally that
\begin{equation}
\label{eq:app_generic_limits}
\begin{aligned}
    \widehat\theta_M
    &\xrightarrow{\mathrm{a.s.}}
    \theta^\dagger
    =
    \begin{pmatrix}
        m_1^\dagger\\
        m_3^\dagger
    \end{pmatrix},\\
    \widehat\Sigma_{\mathrm{BM},M}
    &\xrightarrow{\mathrm{a.s.}}
    \Sigma\succ0.
\end{aligned}
\end{equation}
For the online ellipsoidal rule, also assume
$\widehat\Lambda_{\mathrm{online},N} \xrightarrow{\mathrm{a.s.}}\Lambda\succ0$. For the fixed geometry variant, condition on the positive- definite matrix $\Lambda_{\mathrm{cal}}$. Because the critical values remain bounded and the confidence-region diameter is of order $M^{-1/2}$,
\begin{equation}
\label{eq:app_ruleI_limits}
\begin{aligned}
    D_{1,M}+h_{1,M}
    &\xrightarrow{\mathrm{a.s.}}
    |m_1^\dagger-C|,\\
    D_{3,M}+h_{3,M}
    &\xrightarrow{\mathrm{a.s.}}
    |m_3^\dagger|,
\end{aligned}
\end{equation}
and
\begin{equation}
\label{eq:app_ruleII_limit}
    R_{\Lambda,M}
    \xrightarrow{\mathrm{a.s.}}
    \|\theta^\dagger-\theta^\star\|_{\Lambda^{-1}}.
\end{equation}
Consequently, if the limiting monitored moments lie strictly inside the corresponding tolerance region, the relevant rule eventually passes at every sufficiently late checkpoint almost surely. If they lie strictly outside, it eventually fails at every sufficiently late checkpoint. The boundary case is not decided by this limit argument. In particular, under Eq.~\eqref{eq:app_ergodic_clt}, $\theta^\dagger=\theta^\star$, and any strictly positive tolerances imply eventual almost-sure passage, provided admissible checkpoints continue
indefinitely.

This eventual-passage statement is not a finite-time mixing theorem. At each deterministic large checkpoint, the regions have their usual asymptotic confidence interpretation. Repeatedly inspecting ordinary confidence regions, however, does not make them a time-uniform confidence sequence, and the nominal level need not be preserved exactly at the random stopping time. A literal sequential coverage guarantee would require a confidence sequence, an error-spending construction, or another explicitly time-uniform method. The proposed rules should therefore be interpreted as asymptotically justified MCMC output-analysis procedures for the monitored moments, rather than finite-sample certificates of global mixing or convergence of every observable.

\paragraph{Implementation safeguards.}
The implementation evaluates the rules only every $m$ accepted events and only after a minimum selected-window length, minimum batch size, and minimum number of batches are available. The ellipsoidal additionally requires $a>p$. The online estimate of $\Lambda$ is used only after its own minimum late-window length has been reached.

The theoretical formulas assume that both $\widehat\Sigma_{\mathrm{BM},M}$ and $\Lambda$ are positive definite on the monitored coordinates. Numerically, the implementation first symmetrizes these matrices and, when necessary, applies a small eigenvalue floor before inversion or square-rooting. Such regularization prevents numerical failure but changes the reported confidence and tolerance geometries; its magnitude and the resulting condition numbers should therefore be monitored. If degeneracy is structural rather than numerical, the principled remedy is to reduce the monitored statistic to its nondegenerate subspace rather than rely on an arbitrarily large ridge.


\section{Quantum trajectory simulation}
\label{app:trajectory_simulation}

To reproduce the stochastic trajectories induced by the partial measurements of the quantum Gibbs sampler, we consider two trajectory-sampling procedures: direct trajectory simulation and aggregate trajectory simulation. For the stochastic simulations the latter was the one used.

\paragraph{Direct trajectory simulation.} Let $N\in\mathbb{N}$ be the target number of Kraus applications, let $\ket{\psi_0}$ be a normalized initial state, and let $\{E_k\}_{k=0}^{K}$ be Kraus operators satisfying $\sum_{k=0}^{K} E_k^{\dagger} E_k = I$. A trajectory of length $N$ is a sequence $T=(k_1,\ldots,k_N)$, and its associated unnormalized state is
\begin{equation}
    \ket{\psi_T} = E_{k_N} \cdots E_{k_1} \ket{\psi_0}.
\end{equation}
Its probability is therefore
\begin{equation}
\label{eq:trajectory_probability_appendix}
    \Pr[T] = \norm{\ket{\psi_T}}^2.
\end{equation}
Although the trajectory weights are most naturally expressed in terms of unnormalized states, the algorithms below renormalize the intermediate state after each stochastic update in order to maintain numerical stability.

The direct scheme is the standard step-by-step simulation described in Algorithm~\ref{alg:trajectory_direct_appendix}. At step $n$, the next branch is sampled from the conditional distribution

\begin{equation}
    \Pr[k_n = k \mid k_1,\ldots,k_{n-1}] = \frac{\norm{E_k \ket{\psi_{n-1}}}^2}{\norm{\ket{\psi_{n-1}}}^2},
\end{equation}
and the state is updated accordingly.

\begin{algorithm}[H]
\caption{Direct trajectory simulation}
\label{alg:trajectory_direct_appendix}
\small
\begin{algorithmic}[1]
\Require Normalized initial state $\ket{\psi}$, target length $N$, Kraus operators $\{E_k\}_{k=0}^{K}$
\Ensure Normalized final state $\ket{\psi}$ and trajectory $T=(k_1,\ldots,k_N)$
\State $T \gets ()$
\For{$n=1,\ldots,N$}
    \For{$k=0,\ldots,K$}
        \State $p_k \gets \norm{E_k\ket{\psi}}^2$
    \EndFor
    \State Sample $k_\star \in \{0,\ldots,K\}$ according to $\{p_k\}_{k=0}^{K}$
    \State $\ket{\psi} \gets E_{k_\star}\ket{\psi} / \norm{E_{k_\star}\ket{\psi}}$
    \State Append $k_\star$ to $T$
\EndFor
\State \Return $(\ket{\psi}, T)$
\end{algorithmic}
\end{algorithm}

\paragraph{Aggregate trajectory simulation.} The aggregate scheme is advantageous when one branch, taken to be $k=0$, dominates the evolution. More precisely, assume that there exists $\delta\in(0,1)$ such that
\begin{equation}
    \norm{E_0\ket{\phi}}^2 \ge 1-\delta
\end{equation}
for every normalized state $\ket{\phi}$. Then long runs of the dominant branch occur with high probability, since
\begin{equation}
    \norm{(E_0)^n \ket{\phi}}^2 \ge (1-\delta)^n.
\end{equation}
The aggregate algorithm exploits this structure by sampling an entire block of consecutive dominant events in a single subroutine call, followed by one non-dominant event whenever such a transition is available; see Algorithms~\ref{alg:trajectory_aggregate_helper_appendix} and \ref{alg:apply_dominant_branch_appendix}. When powers of $E_0$ can be evaluated efficiently, for example after a spectral decomposition of $E_0$, this can substantially reduce the cost of simulating typical trajectories while preserving the exact trajectory distribution.

\begin{algorithm}[H]
\caption{Aggregate trajectory simulation}
\label{alg:trajectory_aggregate_helper_appendix}
\small
\begin{algorithmic}[1]
\Require Normalized initial state $\ket{\psi}$, target length $N$, Kraus operators $\{E_k\}_{k=0}^{K}$, dominance parameter $\delta\in(0,1)$
\Ensure Normalized final state $\ket{\psi}$ and trajectory $T=(k_1,\ldots,k_N)$
\State $T \gets ()$
\State $n \gets 0$
\State $b \gets \textsc{true}$
\While{$n < N$}
    \If{$b$}
        \State $N_{\mathrm{rem}} \gets N-n$
        \State $(\ket{\psi}, n_d) \gets \Call{ApplyDominant}{\ket{\psi}, E_0, \delta, N_{\mathrm{rem}}}$
        \State Append $n_d$ copies of $0$ to $T$
        \State $n \gets n + n_d$
        \State $b \gets \textsc{false}$
    \Else
        \For{$k=1,\ldots,K$}
            \State $q_k \gets \norm{E_k\ket{\psi}}^2$
        \EndFor
        \State $q \gets \sum_{k=1}^{K} q_k$
        \If{$q = 0$}
            \State $\ket{\psi} \gets E_0\ket{\psi} / \norm{E_0\ket{\psi}}$
            \State Append $0$ to $T$
            \State $n \gets n + 1$
            \State $b \gets \textsc{true}$
        \Else
            \For{$k=1,\ldots,K$}
                \State $p_k \gets q_k / q$
            \EndFor
            \State Sample $k_\star \in \{1,\ldots,K\}$ according to $\{p_k\}_{k=1}^{K}$
            \State $\ket{\psi} \gets E_{k_\star}\ket{\psi} / \norm{E_{k_\star}\ket{\psi}}$
            \State Append $k_\star$ to $T$
            \State $n \gets n + 1$
            \State $b \gets \textsc{true}$
        \EndIf
    \EndIf
\EndWhile
\State \Return $(\ket{\psi}, T)$
\end{algorithmic}
\end{algorithm}

\begin{algorithm}[H]
\caption{Helper routine for the dominant branch}
\label{alg:apply_dominant_branch_appendix}
\small
\begin{algorithmic}[1]
\Require Normalized state $\ket{\psi}$, dominant Kraus operator $E_0$, dominance parameter $\delta\in(0,1)$, remaining length $N_{\mathrm{rem}}$
\Ensure Normalized updated state $\ket{\psi}$ and integer $n_d\in\{0,\ldots,N_{\mathrm{rem}}\}$
\State Sample $r \sim \mathrm{Uniform}(0,1)$
\State $n_0 \gets \left\lceil \frac{\log r}{\log(1-\delta)} - 1 \right\rceil$
\State $n_0 \gets \max\{0, n_0\}$
\If{$n_0 \ge N_{\mathrm{rem}}$}
    \State $\ket{\phi} \gets E_0^{N_{\mathrm{rem}}}\ket{\psi}$
    \State $\ket{\psi} \gets \ket{\phi} / \norm{\ket{\phi}}$
    \State \Return $(\ket{\psi}, N_{\mathrm{rem}})$
\EndIf
\State $\ket{\phi} \gets E_0^{n_0}\ket{\psi}$
\State $n_d \gets n_0$
\While{$n_d < N_{\mathrm{rem}}$ \textbf{and} $\norm{E_0\ket{\phi}}^2 > r$}
    \State $\ket{\phi} \gets E_0\ket{\phi}$
    \State $n_d \gets n_d + 1$
\EndWhile
\State $\ket{\psi} \gets \ket{\phi} / \norm{\ket{\phi}}$
\State \Return $(\ket{\psi}, n_d)$
\end{algorithmic}
\end{algorithm}

\paragraph{Correctness of the aggregate scheme.}
Let $S=(k_1,\ldots,k_m)$ be a trajectory prefix and define
\begin{equation}
    \ket{\psi_S} = E_{k_m} \cdots E_{k_1} \ket{\psi_0},
\end{equation}
with the convention $\ket{\psi_{\varnothing}} = \ket{\psi_0}$. Let $p(S)$ denote the probability that the sampled trajectory begins with $S$, and set $p(\varnothing)=1$. We claim that
\begin{equation}
    p(S) = \norm{\ket{\psi_S}}^2
\end{equation}
for every prefix $S$, which in particular yields Eq.~\eqref{eq:trajectory_probability_appendix}. The proof proceeds by induction on the prefix length. For the empty prefix $S=\varnothing$, we have $p(\varnothing)=1$ and $\norm{\ket{\psi_\varnothing}}^2=\norm{\ket{\psi_0}}^2=1$, since $\ket{\psi_0}$ is normalized.

Assume that the claim holds for a given prefix $S$, and let $\ket{\phi}$ be an arbitrary unnormalized state. Define
\begin{equation}
    P_n(\phi) := \frac{\norm{E_0^n \ket{\phi}}^2}{\norm{\ket{\phi}}^2}.
\end{equation}
The helper routine in Algorithm~\ref{alg:apply_dominant_branch_appendix} draws $r\sim \mathrm{Uniform}(0,1)$ and returns the smallest integer $n_d$ such that $P_{n_d+1}(\phi) \le r < P_{n_d}(\phi)$, truncated at the remaining horizon if necessary. Consequently,
\begin{equation}
    \Pr[\text{at least $n$ additional dominant steps} \mid \phi] = P_n(\phi)
\end{equation}
for every admissible $n$.

Now let $S$ be any prefix. The probability that the trajectory extends from $S$ by one additional dominant event is
\begin{equation}
    p(S,0) = p(S) \frac{\norm{E_0 \ket{\psi_S}}^2}{\norm{\ket{\psi_S}}^2} = \norm{E_0 \ket{\psi_S}}^2.
\end{equation}
This identity remains valid even when $S$ already terminates in dominant events, since the helper routine depends only on the current state $\ket{\psi_S}$ and not on how that state was reached.

For a non-dominant branch $k\in\{1,\ldots,K\}$, two cases may occur. If
\begin{equation}
    \sum_{j=1}^{K} \norm{E_j\ket{\psi_S}}^2 = 0,
\end{equation}
then no non-dominant continuation is available from the current state. Hence
\begin{equation}
    p(S,k)=0=\norm{E_k\ket{\psi_S}}^2
\end{equation}
for every $k\in\{1,\ldots,K\}$. Otherwise, the aggregate algorithm first
conditions on the event that the dominant block terminates at the current
prefix and then samples among the branches $1,\ldots,K$ with probabilities
proportional to $\norm{E_k\ket{\psi_S}}^2$. Therefore,
\begin{equation}
    p(S,k) =
    \frac{\norm{E_k \ket{\psi_S}}^2}{\sum_{j=1}^{K} \norm{E_j \ket{\psi_S}}^2}
    \bigl(p(S)-p(S,0)\bigr).
\end{equation}
Using $\sum_{j=0}^{K} E_j^{\dagger}E_j = I$, we obtain, whenever
$\sum_{j=1}^{K}\norm{E_j\ket{\psi_S}}^2>0$,
\begin{equation}
    \sum_{j=1}^{K} \norm{E_j \ket{\psi_S}}^2
    = \norm{\ket{\psi_S}}^2 - \norm{E_0 \ket{\psi_S}}^2
    = p(S)-p(S,0),
\end{equation}
so that
\begin{equation}
    p(S,k)=\norm{E_k\ket{\psi_S}}^2
\end{equation}
for every $k\in\{1,\ldots,K\}$. Together with the identity
\begin{equation}
    p(S,0)=\norm{E_0\ket{\psi_S}}^2,
\end{equation}
this proves that
\begin{equation}
    p(S,k)=\norm{E_k\ket{\psi_S}}^2
\end{equation}
for every $k\in\{0,\ldots,K\}$.


\section{Proof of Theorem~\ref{thm:energy_vs_frequencies}}
\label{sec:proof_energy_vs_frequencies}

We prove that the quasi-frequency operator density spans the same operator subspace as the energy projectors, that is,
\[
\vspan\{Q_{\omega}:\omega\in\mathbb R\}
=
\vspan\{\Pi_{\epsilon}:\epsilon\in\operatorname{spec}(H)\}.
\]
The argument proceeds by writing \(Q_{\omega}\) as the image of the energy projectors under a kernel \(G\), and then showing that \(G\) has full column rank.

Using the decomposition
\[
\hat A_a(\omega)=\sum_{\nu\in\mathcal B(H)} \hat A_a(\nu)\,\hat f(\omega-\nu),
\]
we can expand the quasi-frequency operator density as
\begin{equation}
\label{eq:appendix_Qomega_expansion}
\begin{split}
Q_{\omega}
&=
\sum_{a\in\mathcal A}
\sum_{\nu_1,\nu_2\in\mathcal B(H)}
\gamma(\omega)\,
\overline{\hat f(\omega-\nu_1)}\,\hat f(\omega-\nu_2)
\\
&\qquad\times
\sum_{\epsilon\in\operatorname{spec}(H)}
\Pi_{\epsilon-\nu_2}A_a^\dagger \Pi_{\epsilon}A_a\Pi_{\epsilon-\nu_1}.
\end{split}
\end{equation}
with the convention \(\Pi_{\epsilon+\nu}=0\) whenever \(\epsilon+\nu\notin\operatorname{spec}(H)\).

Introduce the completely positive map
\begin{equation}
    \mathcal P^*(X):=\sum_{a\in\mathcal A} A_a^\dagger X A_a,
\end{equation}
and the \(\omega\)-dependent kernel
\begin{equation}
    \Gamma_{\nu_1,\nu_2}(\omega)
    :=
    \gamma(\omega)\,\overline{\hat f(\omega-\nu_1)}\,\hat f(\omega-\nu_2).
\end{equation}
Then \eqref{eq:appendix_Qomega_expansion} becomes
\begin{equation}
\label{eq:appendix_Qomega_compact}
    Q_{\omega}
    =
    \sum_{\nu_1,\nu_2\in\mathcal B(H)}
    \Gamma_{\nu_1,\nu_2}(\omega)
    \sum_{\epsilon\in\operatorname{spec}(H)}
    \Pi_{\epsilon-\nu_2}\,\mathcal P^*(\Pi_{\epsilon})\,\Pi_{\epsilon-\nu_1}.
\end{equation}

Under the uniform prior assumption \eqref{eq:uniform_prior}, we have
\[
\mathcal P^*(X)=\frac{\Tr(X)}{d}\,\bbI.
\]
Substituting this into \eqref{eq:appendix_Qomega_compact}, only the diagonal terms \(\nu_1=\nu_2\) survive, and we obtain
\begin{align}
\label{eq:appendix_Qomega_G}
    Q_{\omega}
    &=
    \sum_{\nu\in\mathcal B(H)} \Gamma_{\nu,\nu}(\omega)
    \sum_{\epsilon\in\operatorname{spec}(H)}
    \frac{\Tr[\Pi_{\epsilon}]}{d}\,\Pi_{\epsilon-\nu} \nonumber \\
    &= 
    \sum_{\epsilon\in\operatorname{spec}(H)} G_{\omega\epsilon}\,\Pi_{\epsilon},
\end{align}
where
\begin{equation}
\label{eq:appendix_G_kernel}
    G_{\omega\epsilon}
    :=
    \sum_{\nu\in\mathcal B(H)}
    \Gamma_{\nu,\nu}(\omega)\,
    \frac{\Tr[\Pi_{\epsilon+\nu}]}{d}.
\end{equation}
This immediately implies
\begin{equation}
\label{eq:appendix_first_inclusion}
    \vspan\{Q_{\omega}:\omega\in\mathbb R\}
    \subseteq
    \vspan\{\Pi_{\epsilon}:\epsilon\in\operatorname{spec}(H)\}.
\end{equation}

To prove the converse inclusion, it suffices to show that the columns \(\omega\mapsto G_{\omega\epsilon}\) are linearly independent, or equivalently, that \(G\) has full column rank. We factor \(G\) as
\begin{equation}
\label{eq:appendix_G_factorization}
    G = S R,
\end{equation}
where
\begin{align}
\label{eq:appendix_RS_def}
    R_{\nu\epsilon}&:=\frac{\Tr[\Pi_{\epsilon+\nu}]}{d}, \\
    S_{\omega\nu}&:=\Gamma_{\nu,\nu}(\omega)
    =\gamma(\omega)\,|\hat f(\omega-\nu)|^2.
\end{align}
The map \(R\) acts from the finite-dimensional energy space to the discrete Bohr-frequency space, while \(S\) acts from the Bohr-frequency space to functions of \(\omega\). The matrix $R$ does not depend on $\gamma$ or $f$, its properties depend only on the energy degeneracies.

We first show that \(R\) has full column rank. Let the distinct energies be ordered as
\[
\epsilon_1<\epsilon_2<\cdots<\epsilon_n,
\qquad
n:=|\operatorname{spec}(H)|.
\]
Consider the \(n\times n\) submatrix of \(R\) obtained by selecting the rows corresponding to the Bohr frequencies
\[
\nu_\ell:=\epsilon_1-\epsilon_\ell,
\qquad
\ell=1,\dots,n.
\]
Its entries are
\[
R_{\nu_\ell,\epsilon_k}
=
\frac{1}{d}\Tr[\Pi_{\epsilon_k+\epsilon_1-\epsilon_\ell}].
\]
If \(\ell>k\), then
\[
\epsilon_k+\epsilon_1-\epsilon_\ell
=
\epsilon_1+(\epsilon_k-\epsilon_\ell)
<\epsilon_1,
\]
so this value is not in the spectrum and the corresponding entry vanishes. Hence the chosen sub-matrix is upper triangular. Its diagonal entries are all equal to
\[
R_{\nu_\ell,\epsilon_\ell}
=
\frac{1}{d}\Tr[\Pi_{\epsilon_1}]>0.
\]
Therefore this sub-matrix is invertible, and \(R\) has full column rank \(n\).

We next show that \(S\) has full column rank under the Gaussian assumptions. Writing
\begin{align}
    \gamma(\omega)=\exp\!\left(-\frac{(\omega+\omega_\gamma)^2}{2\sigma_\gamma^2}\right), \\
    \hat f(\omega)=\frac{1}{\sqrt{\sigma_E\sqrt{2\pi}}}
    \exp\!\left(-\frac{\omega^2}{4\sigma_E^2}\right),
\end{align}
we obtain
\begin{equation}
\label{eq:appendix_S_gaussian}
    S_{\omega\nu}
    =
    \frac{\gamma(\omega)}{\sigma_E\sqrt{2\pi}}
    \exp\!\left(-\frac{(\omega-\nu)^2}{2\sigma_E^2}\right).
\end{equation}
Since \(\gamma(\omega)>0\) for all \(\omega\), multiplication by \(\gamma(\omega)\) does not affect linear independence of the columns. Suppose that
\[
\sum_{\nu\in\mathcal B(H)} c_{\nu}\,
\exp\!\left(-\frac{(\omega-\nu)^2}{2\sigma_E^2}\right)=0
\qquad
\text{for all }\omega\in\mathbb R.
\]
Taking Fourier transforms gives
\[
e^{-\sigma_E^2 t^2/2}
\sum_{\nu\in\mathcal B(H)} c_{\nu} e^{-it\nu}=0
\qquad
\text{for all }t\in\mathbb R.
\]
Since the Gaussian factor never vanishes, we must have
\[
\sum_{\nu\in\mathcal B(H)} c_{\nu} e^{-it\nu}=0
\qquad
\text{for all }t\in\mathbb R.
\]
The exponentials \(e^{-it\nu}\) with distinct \(\nu\) are linearly independent, hence \(c_{\nu}=0\) for all \(\nu\). Therefore the columns of \(S\) are linearly independent, so \(S\) has full column rank.

Since both \(R\) and \(S\) have full column rank, the kernel \(G=SR\) also has full column rank and therefore admits a left pseudo-inverse \(G^+\), with $G^+G=\mathbb{I}$. Applying \(G^+\) to \eqref{eq:appendix_Qomega_G} yields
\[
\Pi_{\epsilon}\in \vspan\{Q_{\omega}:\omega\in\mathbb R\}
\qquad
\text{for every }\epsilon\in\operatorname{spec}(H),
\]
and hence
\begin{equation}
\label{eq:appendix_second_inclusion}
    \vspan\{\Pi_{\epsilon}:\epsilon\in\operatorname{spec}(H)\}
    \subseteq
    \vspan\{Q_{\omega}:\omega\in\mathbb R\}.
\end{equation}
Combining \eqref{eq:appendix_first_inclusion} and \eqref{eq:appendix_second_inclusion} proves the theorem.

\begin{remark}
The uniform prior assumption \eqref{eq:uniform_prior} is stronger than necessary. The proof only uses it to ensure that \(\mathcal P^*(\Pi_{\epsilon})\) stays inside the span of the energy projectors and that the resulting discrete map from energies to Bohr frequencies has full column rank. More generally, the same argument goes through whenever
\[
\mathcal P^*(\Pi_{\epsilon})
=
\sum_{\epsilon'\in\operatorname{spec}(H)}
r_{\epsilon',\epsilon}\,\Pi_{\epsilon'}
\]
for every \(\epsilon\), and the induced discrete matrix
\[
R_{\nu\epsilon}:=r_{\epsilon+\nu,\epsilon}
\]
has full column rank. We state the theorem under the \(1\)-design assumption because it provides a particularly simple, Hamiltonian-agnostic sufficient condition.
\end{remark}

\section{Kraus representation of the QGS channel}
\label{app:kraus_qgs}

We derive an explicit Kraus decomposition for a single step of the quantum Gibbs sampler (QGS) implemented via the weak-measurement circuit shown in Figure~\ref{fig:lindbladian_simulation}. We record the measurement outcomes of the ancillary registers $\rE$, $\rA$, $\rqq$, and $\rqq'$. In the case studies considered in the main text, the prior jumps are proportional to unitaries,
\begin{equation}
    A_a = \frac{1}{\sqrt{|A|}}\, U_a.
\end{equation}
Accordingly, the auxiliary register $\rB$ is unnecessary and the block encoding may be chosen as
\begin{equation}
    V = \sum_{a\in A} \bigl(\ket{a}\!\bra{0} W\bigr)_{\rA} \otimes (U_a)_{\rS},
\end{equation}
where $W$ satisfies
\begin{equation}
    W\ket{0} = \frac{1}{\sqrt{|A|}} \sum_{a\in A} \ket{a}.
\end{equation}

\paragraph{Measurement convention.}
To define a Kraus representation, we fix the following measurement order. We first measure $\rqq'$. An outcome $1$ means a transition event took place. In that case $\rqq$ is necessarily found in $\ket{0}$, and the registers $\rE$ and $\rA$ return the labels $\omega$ and $a$, respectively. If the outcome of $\rqq'$ is $0$, we then perform a joint measurement on the registers $\rE, \rA$, and $\rqq$. The outcome $\ket{000}_{\rE\rA\rqq}$ corresponds to a decay event, whereas every orthogonal outcome is grouped into an error event. For the latter, it is convenient to undo the preparation unitaries $F$ on $\rE$ and $W$ on $\rA$ and record the resulting triple $(t,a,x)$. This merely specifies a particular basis for the ancillary trace and therefore leaves the averaged error channel unchanged.

Let
\begin{equation}
    F\ket{0} = \sum_{t\in S_{t_0}} f(t)\ket{t},
    \qquad
    \qft \ket{\omega} = \frac{1}{\sqrt{N}} \sum_{t\in S_{t_0}} e^{-i \omega t} \ket{t},
\end{equation}
and denote the single-qubit rotation
\begin{equation}
    Y_\theta =
    \begin{pmatrix}
        \sqrt{1-\theta} & -\sqrt{\theta} \\
        \sqrt{\theta}   & \sqrt{1-\theta}
    \end{pmatrix},
    \qquad
    Y(\omega) := Y_{1-\gamma(\omega)}.
\end{equation}
We also set
\begin{equation}
    U_a(t) = e^{i t H} U_a e^{-i t H},
    \qquad
    \hat{U}_a(\omega) = \sum_{t\in S_{t_0}} \frac{e^{-i \omega t}}{\sqrt{N}} f(t)\, U_a(t).
\end{equation}
With this notation, the unitary $U$ of Fig \ref{fig:block_encoding_u} is
\begin{equation}
\begin{split}
    U = \sum_{a\in A} \sum_{\omega \in S_{\omega_0}} \sum_{t\in S_{t_0}}
    &\frac{e^{-i \omega t}}{\sqrt{N}}
    \bigl(\ket{\omega}\!\bra{t} F\bigr)_{\rE}
    \otimes \bigl(\ket{a}\!\bra{0} W\bigr)_{\rA}\\
    &\otimes (U_a(t))_{\rS}
    \otimes (Y(\omega))_{\rqq}.
\end{split}
\end{equation}

For later use, it is convenient to define the projector orthogonal to the all-zero auxiliary subspace
\begin{equation}
    P := \bbI - \ket{000}\!\bra{000}_{\rE\rA\rqq} \otimes \bbI_{\rS},
\end{equation}
and the operators
\begin{equation}
    J := U^{\dagger}
    \bigl(\ket{0}\!\bra{0}_{\rqq} \otimes \bbI_{\rE\rA\rS}\bigr)
    U
    \bigl(\ket{000}_{\rE\rA\rqq} \otimes \bbI_{\rS}\bigr),
\end{equation}
and
\begin{equation}
    \Delta := \sum_{a\in A} \sum_{\omega \in S_{\omega_0}}
    \frac{\gamma(\omega)}{|A|} \, \hat{U}_a(\omega)^{\dagger} \hat{U}_a(\omega).
\end{equation}

With the above measurement convention, one step of the QGS channel admits the following Kraus operators.

\paragraph{Decay term.}
The decay branch corresponds to the outcome $\ket{0000}_{\rE\rA\rqq\rqq'}$:
\begin{equation}
    D = \bra{000}_{\rE\rA\rqq} J
    = \bbI - \bigl(1 - \sqrt{1-\delta}\bigr)\Delta.
\end{equation}

\paragraph{Transition terms.}
The transition branches correspond to the outcomes $\ket{\omega\, a\, 0 1}_{\rE\rA\rqq\rqq'}$:
\begin{equation}
    T_{a,\omega}
    =
    \sqrt{\frac{\delta \gamma(\omega)}{|A|}}\, \hat{U}_a(\omega).
\end{equation}

\paragraph{Error terms.}
The remaining contribution defines the error channel
\begin{equation}
    \rho \mapsto \Tr_{\rE\rA\rqq}\!\left[ P J \rho J^{\dagger} P \right].
\end{equation}
Choosing the basis obtained by applying $(F^{\dagger}\otimes W^{\dagger})_{\rE\rA}$ and then measuring $\rE\rA\rqq$ in the computational basis yields Kraus operators indexed by $(a,t,x)$:
\begin{equation}
\label{eq:error_kraus_qgs}
\begin{split}
    E_{a,t,x}
    =
    \bigl(1-\sqrt{1-\delta}\bigr)
    \Big(
    &\sum_{\omega\in S_{\omega_0}}
    \frac{e^{i \omega t}}{\sqrt{N}} \,
    \sqrt{\frac{\gamma(\omega)}{|A|}} \,
    \mel{x}{Y^{\dagger}(\omega)}{0} \\
    &(U_a(t))^{\dagger} \hat{U}_a(\omega)
    - \delta_{x,0}\,\frac{f(t)}{\sqrt{|A|}} \, \Delta
    \Big).
\end{split}
\end{equation}
For $x = 1$, this simplifies to
\begin{equation}
\begin{split}
    E_{a,t,1}
    =
    -
    \bigl(1-\sqrt{1-\delta}\bigr)
    \sum_{\omega\in S_{\omega_0}}
    \frac{e^{i \omega t}}{\sqrt{N}} \,
    \sqrt{\frac{\gamma(\omega)\bigl(1-\gamma(\omega)\bigr)}{|A|}} \\
    (U_a(t))^{\dagger} \hat{U}_a(\omega),
\end{split}
\end{equation}
whereas for $x=0$ one obtains
\begin{equation}
\begin{split}
    E_{a,t,0}
    =
    \bigl(1-\sqrt{1-\delta}\bigr)
    \bigl(
    \sum_{\omega\in S_{\omega_0}}
    \frac{e^{i \omega t}}{\sqrt{N}} \,
    \frac{\gamma(\omega)}{\sqrt{|A|}} \,
    (U_a(t))^{\dagger} \hat{U}_a(\omega)\\
    - \frac{f(t)}{\sqrt{|A|}} \, \Delta
    \bigr).
\end{split}
\end{equation}

Equivalently, the one-step QGS channel can be written as
\begin{equation}
\begin{split}
    \mathcal{G}(\rho)
    =
    D \rho D^{\dagger}
    + \sum_{a\in A} \sum_{\omega\in S_{\omega_0}}
      T_{a,\omega} \rho T_{a,\omega}^{\dagger} \\
    + \sum_{a\in A} \sum_{t\in S_{t_0}} \sum_{x\in\{0,1\}}
      E_{a,t,x} \rho E_{a,t,x}^{\dagger}.
\end{split}
\end{equation}

\paragraph{Reduced-cost substitute for the error term.}

Computing the exact error contribution is significantly more costly than constructing the leading operators, so we provide a cheaper alternative, which is valid when the error probability is under control. The substitute is constructed to mimic a depolarizing channel with the appropriate normalization to ensure trace preservation. Given an unnormalized channel with Kraus operators $E_k$, we define
\begin{equation}
S = \sum_k E_k^\dagger E_k.
\end{equation}
In our case, the $E_k$ correspond to decay and transition processes. The implemented channel is then
\begin{equation}
R(X) = \Tr[(\bbI - S) X] \, \frac{\bbI}{d},
\end{equation}
and the associated Kraus operators are given by
\begin{equation}
K_k = \sqrt{\frac{w_k}{d}} \, \ketbra{u_l}{u_k},
\end{equation}
where $w_k$ and $\ket{u_k}$ arise from the eigenvalue decomposition of $P = \bbI - S$.


\end{document}